\documentclass[reqno,12pt]{amsart}
\usepackage{amsmath,epsfig}
\usepackage{amssymb}

\newtheoremstyle{rem}{1.3ex}{1.3ex}{\rmfamily}{}
{\itshape\rmfamily}{}{1.5ex}{}

\newtheorem{theorem}{Theorem}[section]
\newtheorem{lemma}[theorem]{Lemma}

\newtheorem{corollary}[theorem] {Corollary}

\theoremstyle{definition}

\newtheorem{remark}[theorem] {Remark}

\renewcommand{\section}{\secdef\sct\sect}
\newcommand{\sct}[2][default]{\refstepcounter{section}
\setcounter{equation}{0}
\vspace{0.5cm}
\centerline{ \large
\scshape \arabic{section}.\ #1}
\vspace{0.3cm}}
\newcommand{\sect}[1]{
\vspace{0.5cm}
\centerline{\large\scshape #1}
\vspace{0.3cm}}

\renewcommand{\subsection}{\secdef \subsct\sbsect}
\newcommand{\subsct}[2][default]{\refstepcounter{subsection}
\nopagebreak
\vspace{0.5\baselineskip}
{\flushleft\bf \arabic{section}.\arabic{subsection}~\bf #1  }
\nopagebreak}
\newcommand{\sbsect}[1]{\vspace{0.1cm}\noindent
{\bf #1}\vspace{0.1cm}}

\def\nn{\mathbb N}

\def\qq{\mathbb Q}

\def\phi{\varphi }

\def\cala{{\mathcal A}}

\def\calc{{\mathcal C}}

\def\calm{{\mathcal M}}
\def\caln{{\mathcal N}}
\def\calo{{\mathcal O}}

\def\cala{{\mathcal A}}

\newcommand{\eps}{\varepsilon}

\newcommand{\supp}{{\operatorname {supp}}}

\newcommand{\R}     {\mathbb{R}}

\newcommand{\N}     {\mathbb{N}}

\newcommand{\E}     {\mathbb{E}}

\def\1{{\mathchoice {1\mskip-4mu\mathrm l}
                    {1\mskip-4mu\mathrm l}
                    {1\mskip-4.5mu\mathrm l} {1\mskip-5mu\mathrm l}}}

\newcommand*{\pe}{\mathbb{P}} 
\newcommand*{\lambdamax}{\lambda^\ast_n}
\newcommand*{\lambdamin}{\lambda^\ast_1}

\newcommand*{\one}{\frac{1}{n}} 
\newcommand*{\spann}{\operatorname{Span}}

\newcommand*{\limsupn}{\limsup\limits_{n \rightarrow \infty}}
\newcommand*{\limn}{\lim\limits_{n \rightarrow \infty}}

\newcommand*{\liminfn}{\liminf\limits_{n \rightarrow \infty}}

\newtheorem{theo}[equation]{Theorem}
\newtheorem{propos}[equation]{Proposition}

\newtheorem{lem}[equation]{Lemma}

\newtheorem{vorbem}[equation]{Note}

\newtheorem{vorremark}[equation]{Remark}

\newcommand{\ba}{\begin{array}}
\newcommand{\ea}{\end{array}}
\newcommand{\be}{\begin{equation}}
\newcommand{\ee}{\end{equation}}
\newcommand{\bea}{\begin{eqnarray}}
\newcommand{\eea}{\end{eqnarray}}
\newcommand{\beas}{\begin{eqnarray*}}
\newcommand{\eeas}{\end{eqnarray*}}

\begin{document}

\title[Large deviations for the largest eigenvalue]{\large
Large deviations for\\\vspace{2mm}the largest eigenvalue of disordered bosons \\\vspace{2mm}and disordered fermionic systems} 

\author[Katrin Credner and Peter Eichelsbacher]{} 
\maketitle
\thispagestyle{empty}
\vspace{0.2cm}

\centerline{\sc Katrin Credner, Peter Eichelsbacher\footnote{Corresponding author: Ruhr-Universit\"at Bochum, Fakult\"at f\"ur Mathematik,
NA 3/66, D-44780 Bochum, Germany, {\tt peter.eichelsbacher@ruhr-uni-bochum.de}
\\All authors have been supported by Deutsche Forschungsgemeinschaft via SFB/TR 12}}
\centerline{\sc (Ruhr-Universit\"at Bochum)}


\vspace{2 cm}

\begin{quote}
{\small {\bf Abstract:} }
We prove a large deviations principle
for the largest eigenvalue of a class of biorthogonal and multiple orthogonal
polynomial ensembles that includes a matrix model of Lueck, Sommers and Zirnbauer for disordered bosons
and Angelesco ensembles. Moreover we consider matrix ensembles in mesoscopic physics.
\end{quote}



\setcounter{section}{0}

\medskip
\section{Introduction}

Muttalib introduced in \cite{Muttalib:1995} a new model of random matrices with two-body interactions. This approach was motivated by physics, where the usage of a standard random matrix ansatz for modelling the behaviour of disordered conductors leads to a small deviation from observed results. Although it was known using methods of perturbation theory that a small correction would be necessary for exact results (see \cite{LeeStone:1985}), it was not known how this model could be solved mathematically. Muttalib started at this point and contributed a solvable random matrix model that considered a small correction of the classical one-body model needed for e.g.~metallic conductors. Muttalib's model has the probability density function
\begin{equation} \label{biorth}
q_n(x_1,\ldots, x_n) = c_n \prod_{1 \leq i<j \leq n} |x_i-x_j| |x_i^{\theta}-x_j^{\theta}| \prod_{i=1}^n e^{-V(x_i)} \,\, 1_{\Sigma^n}(x_1,\ldots,x_n)\, ,
\end{equation}
where $\theta$ is a fixed positive number, the $x_j, \, 1\leq j\leq n$, are the eigenvalues of a $n\times n$ matrix $X$, $V$ is a weight function specifying the concrete model. The fraction $c_n$ normalises the density and $\Sigma$ is the domain of the eigenvalues.
Since density (\ref{biorth}) can be rewritten with the help of biorthogonal polynomials, Borodin \cite{Borodin:1999} introduced the term {\em biorthogonal ensembles} for ensembles with density (\ref{biorth}). Since one has a repulsion between the different eigenvalues $x_i$ and a repulsion between the $x_i^{\theta}$'s, the point
process described in \eqref{biorth} is sometimes called a point process with two-particle interactions. Moreover these point processes are determinantal 
point processes (see \cite[Definition 4.2.11]{Zeitounibook}). 

This paper is motivated by a special biorthogonal ensemble that arises from a random matrix model for {\it disordered bosons} that was proposed by Lueck, Sommers, and Zirnbauer in \cite{Lueck/Sommers/Zirnbauer:2006}. The model amounts to the product of a Wishart matrix and the fundamental matrix of the standard symplectic form, 
and on the level of eigenvalues, interpreted as characteristic frequencies of disordered quasi-particles, one obtains the joint density
\begin{equation} \label{densitybosonic}
{\tilde q}_{n,\alpha}(x_1,\ldots, x_{n})= \frac{1}{Z_{n, \alpha, \tau}} 
\prod_{1\leq i<j \leq n}|x_i-x_j|\left|x_i^{2}-x_j^{2}\right|   \prod_{i=1}^{n}x_i ^{\alpha}\mathrm{e}^{-\tau x_i} \, \,1_{[0,\infty)^n}(x_1,\ldots, x_{n})
\end{equation}
with $\alpha \in \nn \cup \{0\}$ and $\tau >0$. 
In \cite{Lueck/Sommers/Zirnbauer:2006} it was shown
that the correlation functions of the frequencies in the bulk of the spectrum are in the Gaussian Unitary Ensemble universality class, yet a novel
scaling behaviour is found at the low frequency end of the spectrum. Other applications of biorthogonal ensembles to physics are discussed in 
\cite{Tierz:2010} motivated by matrix models for Chern-Simons theory.  For any sequence of random numbers $x_1, \ldots, x_n$ consider
$
L_n : = \frac 1n \sum_{i=1}^n \delta_{x_i},
$
the empirical distribution or {\it empirical measure} of these values (a random probability measure on $\Bbb R$). Define
the mean empirical measure $\bar{L}_n = \E L_n$ by the relation $\langle \bar{L}_n, f \rangle = \E \langle L_n, f \rangle$ for all
continuous and bounded functions $f : \Bbb R \to \Bbb R$. Now in \eqref{densitybosonic} we consider $\tau = n$ and $\alpha \in \nn \cup \{0\}$ be fixed. 
This is the case when the Gaussian random variables
in the random matrix model in \cite{Lueck/Sommers/Zirnbauer:2006} are chosen to be random variables with zero mean and variance $1/n$ .
One result in \cite{Lueck/Sommers/Zirnbauer:2006} is that with $\tau = n$ the sequence
$(\bar{L}_n)_n$ converges weakly to a probability measure on $\Bbb R$ with Lebesgue density
\begin{equation} \label{bosonrho}
\varrho_{\infty} (t) : = \frac{1}{2\pi} (t/b)^{-1/3} \bigl( (1+ \sqrt{1-t^2/b^2})^{1/3} -(1- \sqrt{1-t^2/b^2})^{1/3} \bigr), 
\end{equation}
for $0<t\le b:= 3\sqrt{3}$. Hence one would expect that the largest frequency
$$
x^* := x_n^* := \max_{j=1}^n x_j
$$ 
converges almost surely to the right endpoint $3 \sqrt{3}$ of the support of $\varrho_{\infty}$. One aim of this paper is to complement the 
Lueck-Sommers-Zirnbauer results by a {\it large deviations principle} for the largest characteristic frequency $x^*$. 

Recently, in \cite{ClaeysRomano:2015} it was shown that the biorthogonal polynomials associated to models \eqref{biorth} satisfy a recurrence relation and a Christoffel-Darboux formula if $\theta \in \qq$. Moreover the authors express the equilibrium measure associated to this model. 

Actually, we will study large deviations principles for largest eigenvalues
in a broader framework that encompasses not only ensembles like \eqref{biorth}, but also takes care
of weight functions $w_n$, depending on $n$, and determinantal parts like
$ \prod_{i < j} |x_i-x_j|^{\beta} (x_i^{\theta} - x_j^{\theta})$ for any $\beta >0$. Moreover we  consider random matrix ensembles in mesoscopic physics subsuming matrix versions of {\it all classical symmetric spaces}. We will also obtain large deviations results
for {\it multiple orthogonal polynomial} ensembles (see \cite{Kuijlaars:2010}). 

In \cite{BenArous/Dembo/Guionnet:2001}, Ben Arous, Dembo and Guionnet proved a large deviations result for the GOE ensemble where
the joint probability density of the eigenvalues is given by
\begin{equation}
q_n(x_1,\ldots, x_n) = c_n \prod_{1 \leq i<j \leq n} |x_i-x_j| \prod_{i=1}^n e^{- \frac n4 x_i^2}.
\end{equation}
Here the sequence of largest eigenvalues $(x_n^*)_n$ fulfils a LDP in $\R$ with speed $n$ and good rate function
\bea
\nonumber J(x)& =& \left\{\begin{array}{l@{\quad :\quad}l} \int_2^x \sqrt{(t/2)^2-1}\, dt
 & x\geq 2 \\ \infty &x< 2\end{array}\right.\, \, .
\eea
Recall that a family of probability measures $(\mu_{\varepsilon})_{\varepsilon >0}$ on a topological space
$X$ is said to obey a large deviations principle (LDP) with speed $\varepsilon^{-1}$ and good rate function $I: X \to [0,\infty]$ if $I$ is lower semi-continuous
and has compact level sets $N_L := \{x \in X: I(x) \leq L \}$, for every $L \in[0, \infty)$, and
$$
\liminf_{\varepsilon \to 0} \varepsilon \log \mu_{\varepsilon}(G) \geq - \inf_{x \in G} I(x)
$$
for every open $G \subseteq X$ and
$$
\limsup_{\varepsilon \to 0} \varepsilon \log \mu_{\varepsilon}(A) \leq - \inf_{x \in A} I(x)
$$
for every closed $A \subseteq X$.
The result in \cite{BenArous/Dembo/Guionnet:2001} has been generalized in \cite{Zeitounibook} to joint densities of the form 
\begin{equation} \label{hamiltonian}
q_n(x_1,\ldots, x_n)= c_n  \prod_{1\leq i<j \leq n}
\left|x_i-x_j\right|^{\beta} e^{-n \sum_{i=1}^n V(x_i)} , 
\end{equation}
for $\beta >0$, partition function $c_n^{-1}$ and continuous weight functions $V$ which satisfies $\liminf_{|x| \to \infty} V(x) / (\beta' \log |x|) >1$
for some $\beta' >1$ with $\beta' \geq \beta$ (see Theorem 2.6.6 in  \cite{Zeitounibook}).

Recently, in \cite{Fanny:2015} the author proves a large deviations principle for the largeste eigenvalue of Wigner matrices without Gaussian tails, namely such that the distrution tails $P(|X_{1,1}| >t)$ and $P(|X_{1,2}| >t)$ behave like $e^{-b t^{\alpha}}$ and $e^{-a t^{\alpha}}$ respectively for some $a,b \in (0, \infty)$ and
$\alpha \in (0,2)$. The large deviations principle is of speed $n^{\alpha/2}$ and with an explicit good rate function depending only on the tail distributions of the
$X_{i,j}$.

The paper is organised as follows. Section 2 is devoted to the formulation of a large deviations principle for the largest eigenvalue
$x_n^*$ of generalised biorthogonal matrix ensembles. The examples presented in Section 3 include the random matrix
model of disordered bosons in \cite{Lueck/Sommers/Zirnbauer:2006}, biorthogonal Laguerre ensembles 
considered in \cite{Borodin:1999} as well as Wigner-Dyson ensembles, Bogoliubov-de Gennes ensembles and Chiral ensembles.
In Section 4 we formulate large deviations
principles for a special multiple orthogonal ensemble, the Angelesco ensemble, see \cite{Kuijlaars:2010}.
In Sections 5 we present the proofs of our large deviations principles.

\section{Large deviations for $\lambda_n^*$ of biorthogonal ensembles}

In this section, we will derive a LDP for the bosonic ensemble, where 
the density of the joint distribution of the eigenvalues is of form \eqref{densitybosonic}.
Obviously, \eqref{densitybosonic} is a special case of the density
\begin{equation}\label{density}
q_n(\lambda_1,\ldots, \lambda_{p(n)})=\frac{1}{Z_n} \prod_{1\leq i<j \leq p(n)}|\lambda_i-\lambda_j|
\left|\lambda_i^{\theta}-\lambda_j^{\theta}\right|  \prod_{i=1}^{p(n)}w_n(\lambda_i)^n \,\, 1_{\Sigma^{p(n)}}(\lambda_1,\ldots, \lambda_{p(n)}),
\end{equation}
with $\theta \in \N$, partition function $Z_n$ and continuous weight functions $w_n:\R\rightarrow\R_0^+$. For $\theta$ even, 
$\Sigma$ is a closed subset of $[0,\infty)$ while for $\theta$ odd, $\Sigma$ is a closed subset of $\R$. The 
sequence $(p(n))_n$ must satisfy
$$
\lim_{n\rightarrow \infty} \frac{p(n)}{n}=\kappa \in (0,\infty).
$$
Note that for $p(n)=n$ and weight functions $w$ independent of $n$, \eqref{density} subsumes the density for the eigenvalue distribution 
for biorthogonal ensembles as introduced in \cite{Borodin:1999}. For $\theta =1$, the density \eqref{density} is the same as the density
considered in \cite{Eichelsbacher/Stolz:2006} (see formula (4.1) in \cite{Eichelsbacher/Stolz:2006} with $\gamma =1$ and $\beta =2$). 
For $\theta=1$, $p(n)=n$ and $w(x)=\mathrm{e}^{-\frac{1}{2}x^2}$ we 
also recover the classical GUE. 

Throughout the whole section, we write ${\mathcal N}(f)$ for the set of zeros of a function $f:\R\rightarrow \R$ and we assume that the sequence 
of weight functions $(w_n)_n$ satisfies the following:
\begin{itemize}
\item[(a1)] there exists a continuous function $w: \Sigma \to [0,
\infty )$ such that 
\begin{itemize}
\item $\# \caln(w) < \infty,\ \caln(w_n) \subseteq \caln(w)$ for 
large $n$. (a1.1) 
\item As $n \to \infty$, $w_n$ converges to $w$, and $\log w_n$ to
$\log w$ uniformly on compact sets. (a1.2)  
\item  $\log w$ is Lipschitz on compact sets away from $\caln(w)$. (a1.3)
\end{itemize}
\item[(a2)] If $\Sigma$ is unbounded, then there exists $n_0 \in
\nn$ such that
$$ \lim_{x \to \pm \infty} |x|^{(\theta+1) (\kappa + \epsilon)} \sup_{n \ge n_0} w_n(x) = 0$$
for some fixed $\epsilon > 0$.
\end{itemize}
Moreover we assume that the partition functions satisfy
\begin{equation} \label{assZ}
\lim_{n \to \infty} \frac 1n \log \frac{Z_{n-1}}{Z_n} =: \xi,
\end{equation}
where $\xi$ is a constant.

\begin{remark} \label{correction}
We know from \cite[Theorem 2.1]{Eichelsbacher/Sommerauer/Stolz:2011}, that the empirical measure $L_n=\one\sum\limits_{j=1}^n\delta_{\lambda_j}$ of the eigenvalues of a biorthogonal random matrix satisfying (a1) and (a2) obeys a large deviations principle with speed $n^2$ and a good rate function. In some examples
the function  $w_n$ will be of the form $w_n(x) = x^{\frac{\alpha}{n}} \phi_n(x)$ with a fixed $\alpha \geq 0$ and some $\phi_n(x)$. But with the first factor $x^{\frac{\alpha}{n}}$,
$\log w_n$ does not convergence uniformly to some $\log w$ on compact sets. This problem, which arises in Examples 2.3 and 2.5 in \cite{Eichelsbacher/Sommerauer/Stolz:2011}, can be fixed easily. We rewrite \eqref{density} as
\begin{equation}\label{density2}
q_n(\lambda_1,\ldots, \lambda_{p(n)})=\frac{1}{Z_n} \prod_{1\leq i<j \leq p(n)}|\lambda_i-\lambda_j|
\left|\lambda_i^{\theta}-\lambda_j^{\theta}\right|  \prod_{i=1}^{p(n)} \lambda_i^{\alpha} \phi_n(\lambda_i)^n \,\, 1_{\Sigma^{p(n)}}(\lambda_1,\ldots, \lambda_{p(n)}).
\end{equation}
Then we obtain that the empirical measure $L_n=\one\sum\limits_{j=1}^n\delta_{\lambda_j}$ of the eigenvalues of this biorthogonal random matrix, where
$\phi_n$  and a limiting function $\phi$ satisfy (a1) and (a2), obeys a large deviations principle with speed $n^2$ and a good rate function. The proof is step by step
the proof of \cite[Theorem 2.1]{Eichelsbacher/Sommerauer/Stolz:2011}.
\end{remark}

\begin{remark}
In the proof of our main theorem, one basic step is to integrate out the density of all eigenvalues with respect to one and to rewrite some of the remaining parts to a density that belongs to an eigenvalue distribution that misses the integrated eigenvalue.
Since we do not necessarily have that $p(n)=n$, we should clarify what we mean by having ``one eigenvalue less'': We want to analyse $q_n(\lambda_1,\ldots,\lambda_{p(n)})$, which is the joint probability distribution of the $p(n)$ eigenvalues of an $n\times n$ matrix. We do so by inserting the probability distribution of $p(n)-1$ eigenvalues. This is not necessarily $q_{n-1}$, but, since $(p(n))_n$ is a subsequence of the sequence of natural numbers $(n)_n$, there is a $p(n-k)$ with $n>k\in\N$ and $p(n-k)=p(n)-1$. For practical reasons we assume that we have $k=1$, i.e.~that we have to go back in the sequence $(n)_n$ by just one step to have one eigenvalue less. This implies that we have $p(n-1)=p(n)-1$. Further, we denote with (with a slight abuse of notation) $q_{n-1}(\lambda_1,\ldots,\lambda_{p(n)-1})$ the eigenvalue distribution of $p(n)-1$ eigenvalues.
We also have to take care of the normalisation constants of the two densities. Their fraction, logarithmised and divided by $n$, should converge to a constant (see \eqref{assZ} above). Without restriction, we omit here the first eigenvalue $\lambda_1$. We also replace, with a slight abuse of notation (since, in Equation (\ref{5.67}), $w_{p(n)}$ and $w_{p(n)-1}$ would be more correct), $w_n$ by $w_{n-1}$. This last replacement is technically not necessary; we are at freedom to use any transformation of $w_n$ that leads to an exponentially equivalent density of the eigenvalues, as long as it holds that 
\be
\label{5.67} \limn \frac{w_n(x)}{w_{n-1}(x)}=1\, .
\ee
We obtain
\be 
Z_{n-1} = \int_{\Sigma^{p(n)-1}}\prod\limits_{j=2}^{p(n)} (w_{n-1}(\lambda_j))^{n-1} \prod\limits_{2\leq i<j\leq p(n)} |\lambda_i-\lambda_j| |\lambda_i^{\theta}-\lambda_j^{\theta}|\, \prod_{j=2}^{p(n)} d\lambda_j\, .
\ee
To be clear: we have some freedom of choice for the functions $w_n$, and, to resume the discussion above, it is also possible to use $w_{n-k}$ in the normalisation constant above. Since the sequence of the $(w_n)_n$ converges and since we are only interested in the logarithmised and scaled fraction of $Z_{n-1}$ and $Z_n$, this does not make any difference. With the same reasoning, the exponent $n-1$ of $w_{n-1}$ could also be e.g.~$n-k$, with no change in the outcome.
\end{remark}
We will study the asymptotic behaviour of the largest eigenvalue $\lambda_n^*:= \max_{j=1}^{p(n)} \lambda_j$ of $\lambda$ 
for $\lambda = (\lambda_1, \ldots, \lambda_{p(n)}) \in \Sigma^{p(n)}$. The main theorem now reads as follows:

\begin{theo}\label{maintheoremlambdamax}
Let $\lambda_1,\ldots, \lambda_{p(n)}$ be the eigenvalues of a biorthogonal ensemble, that is, with joint eigenvalue density (\ref{density}). 
Under the assumptions formulated above ((a1), (a2) and \eqref{assZ}) the sequence  $(\lambda_n^*)_n$ satisfies a large deviations principle in 
$\Sigma$ with speed $n$ and good rate function
\bea
\nonumber I(x)& =& \left\{\begin{array}{l@{\quad :\quad}l}- \kappa  \int \bigl( \log|x-y| +\log |x^{\theta}-y^{\theta}| \bigr) 
d\mu_w(y) -\log w(x) - \zeta & x\geq b_w \\ \infty &x< b_w\end{array}\right.\, ,
\eea
where
\be
\nonumber\zeta:=\kappa\int\log w(y)d\mu_w(y)+\xi\, ,
\ee
$\xi$ as defined in (\ref{assZ}), $\mu_w$ is the limiting measure of the empirical measure of the eigenvalues $\lambda_1,\ldots, \lambda_{p(n)}$, and $b_w$ the right endpoint of its support. Note that $\zeta$ does not depend on $x$.
\end{theo}

Note that Theorem 2.1 from \cite{Eichelsbacher/Sommerauer/Stolz:2011} ensures that the limiting measure $\mu_w$ exists.
\bigskip

\begin{corollary}\label{corrlambdamax}
Consider the slightly different density
\begin{equation}\label{density2}
q_n(\lambda_1,\ldots, \lambda_{p(n)})=\frac{1}{Z_n} \prod_{1\leq i<j \leq p(n)}
\left|\lambda_i^{\theta}-\lambda_j^{\theta}\right|^{\beta}  \prod_{i=1}^{p(n)}w_n(\lambda_i)^n \,\, 1_{\Sigma^{p(n)}}(\lambda_1,\ldots, \lambda_{p(n)}),
\end{equation}
with $\beta >0$, $\theta \in \N$. Under the assumptions (a1), (a2) and \eqref{assZ} the sequence  $(\lambda_n^*)_n$ satisfies a large deviations principle in 
$\Sigma$ with speed $n$ and good rate function
\bea
\nonumber I(x)& =& \left\{\begin{array}{l@{\quad :\quad}l}- \kappa \beta \int \log |x^{\theta}-y^{\theta}| d\mu_w(y) -\log w(x) - \zeta & x\geq b_w \\ \infty &x< b_w\end{array}\right.\, ,
\eea
where
\be
\nonumber\zeta:=\kappa\int\log w(y)d\mu_w(y)+\xi\, ,
\ee
$\xi$ as defined in (\ref{assZ}), $\mu_w$ is the limiting measure of the empirical measure of the eigenvalues $\lambda_1,\ldots, \lambda_{p(n)}$, and $b_w$ the right endpoint of its support. 
\end{corollary}

\section{Examples}
\subsection{Disordered bosons}
Returning to the bosonic ensemble with density \eqref{densitybosonic}, we have as weight functions
$
w_n(x)=x^{\frac{\alpha}{n}}e^{-\frac{\tau x}{n}} =: x^{\frac{\alpha}{n}} \, \phi_n(x) 
$ 
with a fixed $\alpha \in \nn \cup \{0\}$. Remark that $\tau^{-1}$ is the variance of the independent and normally distributed random variables, that were used to 
construct the stability matrix $h$ for that ensemble, cf. \cite{Lueck/Sommers/Zirnbauer:2006}. We have to take the variance $\tau^{-1}$ equal to 
$n^{-1}$ to be able to obtain a limit of the empirical measures of the eigenvalues.
Obviously, conditions (a1) and (a2) and \eqref{5.67} are met for $\phi_n(x)$ and $\phi(x)= e^{-x}$, see Remark \ref{correction}. We now need to verify \eqref{assZ}. It is far from trivial to calculate
$Z_{n, \alpha, \tau}$ for this matrix ensemble. 

For $\tau=1$ and  $\eta_i(\lambda_j)= \lambda_j^{i-1}$, $\xi_i(\lambda_j) = \lambda_j^{2(i-1)}$ and $w(\lambda_j) = \lambda_j^{\alpha} e^{-\lambda_j}$, we
obtain for $h_j$ defined in the Appendix, that for $j \geq 1$
$$
h_j\delta_{i,j} =\int_I p_i(\lambda)p_j(\lambda) w(\lambda) \, d\lambda = 2^j j! (2j+\alpha)! ,
$$
where the last equality was calculated in \cite[Equation (5.10)]{Lueck/Sommers/Zirnbauer:2006}.
Thus applying \eqref{clever}, we get for the partition function that
$$
Z_{n,\alpha,1} = \hat{c}\cdot n!  \prod_{j=0}^{n-1} h_j = \hat{c}\cdot n! \prod_{j=0}^{n-1}2^j j! (2j+\alpha)!
= \hat{c} \cdot 2^{n(n-1)/2} \prod_{j=1}^n j! (2(j-1)+\alpha)!\, .
$$
We substitute in the former $y_j = \tau \lambda_j$ for all $j=1,\ldots,n$ and obtain (with $\lambda_j=y_j/\tau$ and $d\lambda_j=1/\tau \, dy_j$)
\bea
\nonumber Z_{n,\alpha,\tau} &=& \int \cdots \int \prod_{1\leq i<j\leq n}(\lambda_i-\lambda_j)(\lambda_i^2-\lambda_j^2) \prod_{j=1}^n \lambda_j^{\alpha} e^{-\tau \lambda_j}\, d\lambda_j\\
\nonumber &=& \int \cdots \int \prod_{1\leq i<j\leq n} \frac{1}{\tau}(y_i-y_j)\frac{1}{\tau^2}(y_i^2-y_j^2) \prod_{j=1}^n \tau^{-(\alpha+1)} y_j^{\alpha} e^{-y_j}\, dy_j\\
\nonumber &=& \tau^{-3/2n(n-1)} \tau^{-n(\alpha+1)} Z_{n,\alpha,1} = \tau^{-n(3/2(n-1)+(\alpha+1))} 2^{n(n-1)/2}\prod_{j=1}^n j! (\alpha+2(j-1))!\, ,
\eea
Hence we have proved that
\be
\label{LSZ2006-z_n} Z_{n,\alpha,\tau} = c\cdot \tau^{-3/2n(n-1)-n(\alpha+1)}\cdot 2^{n(n-1)/2} \prod_{j=1}^n j!(\alpha+2(j-1))!\, .
\ee
This gives
\bea
\nonumber \limn\one\log\frac{Z_{n-1}}{Z_n} &=& \limn\one\log \left(  \tau^{3(n-1)+(\alpha+1)}\cdot 2^{-(n-1)}\cdot \frac{1}{n!(\alpha+2(n-1))!}\right)\\
\nonumber&=& \limn\frac{3(n-1)+(\alpha+1)}{n}\log\tau - \limn \frac{n-1}{n}\log 2\\
&& \qquad \hspace{16ex} - \limn\one\log(n!(\alpha+2(n-1))!)\, . \label{5.40} 
\eea
If we choose $\tau =n$, we obtain for equation (\ref{5.40})
$$
 \limn\one\log\frac{Z_{n-1}}{Z_n} = 3(1-\log 2)\, .
$$
Therefore \eqref{assZ} is fulfilled. We may apply Theorem \ref{maintheoremlambdamax} on this model.
For the rate function, we compute the integral
$$
I(x)= - \int \bigl( \log(x-y) +\log (x^2-y^2) \bigr) d\varrho_{\infty}(y) -\log w(x) - \left(\int\log w(y)d\varrho_{\infty}(y)+\xi\right)\, , \nonumber\\\label{5.39}
$$
where $\varrho_{\infty}$ is given by \eqref{bosonrho}. With $\xi = 3-3\log 2$ and $\int_0^{3\sqrt{3}} y\, d\varrho_{\infty}(y) = \frac{3}{2}$ the rate function becomes
$$
I(x) =  - \int \bigl( \log(x-y) +\log (x^2-y^2) \bigr) d\varrho_{\infty}(y) +x -\left(\frac{3}{2} + 3-3\log 2 \right)\, .
$$

\subsection{Laguerre biorthogonal ensembles}
The Laguerre ensembles are a generalisation of the biorthogonal ensemble introduced by L\"uck, Sommers, and Zirnbauer. Take
\be
\nonumber \Sigma = (0,\infty)\;\mbox{ and }\; w_n(x)=x^{\frac{l}{n}}e^{\frac{-\tau x}{n}}\, ,
\ee
with parameter $l\in \N_0$ and with $\tau=n$ and take \eqref{density} with $\theta \in \N$, $\Sigma = \R$ and $p(n)=n$. 
Applying the machinery of biorthogonal polynomials (see \eqref{clever} in the subsection before and the Appendix), we obtain $h_j = j! \theta^j (\theta j + l)!$
and hence $Z_{n,1,l} = c n! \prod_{k=0}^{n-1} k! \theta^k (\theta k+l)!$ and 
$$
Z_{n,\tau,l} = c \cdot \tau^{-\frac{n(n-1)}{2}(\theta+1)-n(l +1)} n!\prod_{k=0}^{n-1} k! \theta^k (\theta k+l)!\, .
$$
If the parameter $l$ is constant, we get, with $\tau = n$ and the use of Stirling's formula that
$$
\limn \one \log \frac{Z_{n-1,n,l}}{Z_{n,n,l}} = \theta +1 -\log \theta\, .
$$
This is constant and therefore Assumption \ref{assZ} is fulfilled. The sequence of weight functions $w_n(\lambda)= \lambda^{l/n}e^{-\lambda}=:  \lambda^{l/n} \phi(\lambda)$ converges for $n\to\infty$ to $w(\lambda) = \phi(\lambda)$ and, as we have already seen, $\phi$ fulfils (a1) and (a2), see Remark \ref{correction}. Therefore, we can apply Theorem \ref{maintheoremlambdamax} to the biorthogonal Laguerre ensembles with constant parameter $l$. That is, we have a large deviations principle with speed $n$ for the largest eigenvalue. We omit the calculation of the rate function. If  $\tau=n$ and $l:=l(n)$ depends on $n$ such that
\be
\nonumber \limn \frac{l(n)}{n} =: L \in (0,\infty)\, ,
\ee
where $L$ is a positive constant, the weight function $w_n$ converges for $n\to\infty$ to $\lambda^L e^{-\lambda} =: w(\lambda)$.
Here we obtain
$$
\limn\one\log \frac{Z_{n-1,n,l}}{Z_{n,n,l}} = \theta + L +1 - \log \theta\, ,
$$
which is constant and therefore fulfils Assumption \ref{assZ}. The assumptions on the weight functions are fulfilled as well.
Therefore, since all assumptions of Theorem \ref{maintheoremlambdamax} are fulfilled, we have a large deviations principle with speed $n$ for the largest eigenvalue of the biorthogonal Laguerre ensembles.
In \cite{Borodin:1999}, Borodin mentioned two more prominent classes of random matrices: the biorthogonal versions of Jacobi and Hermite ensembles. 
The calculation for these ensembles should be similar to the case we presented, which in turn is again based on the calculations in 
\cite{Lueck/Sommers/Zirnbauer:2006}.

\subsection{The tenfold way}
Joint densities defined in \eqref{density2} occure is the framework of mesoscopic physics, since it subsumes matrix versions of all (ten) classical symmetric spaces, see \cite{Eichelsbacher/Stolz:2006}. It can be interpreted as the symmetry classification of disordered {\it fermionic} systems. An analogous classification for the case of bosons is not completely understood, see, however, the discussion in \cite[Section 4]{Zirnbauer:2010}. For the three classical
{\it Wigner-Dyson ensembles} we choose $\beta>0$, $\theta=1$, $w_n(x) =w(x) = e^{-\beta x^2/4}$, $p(n)=n$ and hence $\kappa=1$. Obviously assumptions (a1)
and (a2) are fullfilled. Applying Selberg's integral (see \cite[(17.6.7)]{Mehta:book}), one obtains
$$
Z_n= (2 \pi)^{n/2} \biggl( \frac{\beta n}{2} \biggr)^{- \beta n(n-1)/4 - n/2} \prod_{i=1}^n \frac{\Gamma(\frac{j\beta}{2})}{\Gamma(\frac{\beta}{2})}.
$$
With Stirling's formula we obtain
$$
\lim_{n \to \infty} \frac 1n \log \frac{Z_{n-1}}{Z_n} = -\frac{\beta}{4}.
$$ 
Hence assumption \eqref{assZ} is fulfilled and we obtain the LDP for the largest eigenvalue, first proved in \cite{BenArous/Dembo/Guionnet:2001}
for $\beta=1$. The case $\beta=2$ and $\beta=4$ are included in \cite[Theorem 2.6.6]{Zeitounibook}, but the constant $\xi$ was not calculated explicitly in \cite{Zeitounibook}.

 \subsubsection{Bogoliubov-de Gennes ensembles}\label{section-bdg}
In this section, we prove the large deviations principle stated in Theorem \ref{maintheoremlambdamax} for four (of five) Bogoliubov-de Gennes (BdG) ensemble
(compare with Section 4 and the classification table in Section 3 of \cite{Eichelsbacher/Stolz:2006}).
We choose $\theta=2$, $p(n)=n$ and $w_n(x) = x^{\alpha/n} \exp( - \frac{1}{\psi_{\mathcal C} \sigma^2} x^2 )$, where for the four different classes 
${\mathcal C}$
we choose $\alpha=2$, $\beta=2$ and $\psi_{\mathcal B}=2$ for the class B,  $\alpha=0$, $\beta=2$ and $\psi_{\mathcal D}=2$ for the class D,
$\alpha=2$, $\beta=2$ and $\psi_{\mathcal C}=4$ for the class C and $\alpha=1$, $\beta=1$ and $\psi_{\mathcal CI}=4$ for the class CI. Hence
\eqref{5.67} is fulfilled.
We will prove that the rate function belonging to Corollary \ref{corrlambdamax} for these BdG ensembles is
\be
\label{bdg-ratefunction} I(x) = \left\{\begin{array}{l@{\quad :\quad}l} \beta\frac{4}{b_w^2} \int\limits_{b_w}^x  \sqrt{t^2-b_w^2}\, dt & x\geq b_w\\ \infty &x< b_w\end{array}\right. \, ,
\ee
with $b_w = \sqrt{2\psi_{\calc}\sigma^2\beta \kappa}> 0$ being the right endpoint of the support of the limiting measure $\mu_w$ of the empirical eigenvalue distribution. This limiting measure is defined as 
\bea
\label{mu_w_bdg} \mu_w(y) &=& \frac{2}{\psi_{\calc}\sigma^2\beta\kappa \pi}\sqrt{2\psi_{\calc}\sigma^2\beta\kappa - y^2}\cdot 1_{[0,\sqrt{2\psi_{\calc}\sigma^2\beta \kappa}]}(y)\, .
\eea
Note that since $\theta$ is even, $\Sigma\subseteq [0,\infty)$. Therefore, we need to deal only with positive eigenvalues; which leads to integrating over only positive values. We obtain $I(x)=0$ if and only if $x = b_w$. From the upper bound of the LDP and the Borel-Cantelli lemma it follows that
$
P \bigl( \lambda_n^* \to b_w \bigr) =1$, 
which is a strong law of large numbers. 

\begin{lem}\label{bdg-zn}
For all BdG ensembles, the scaled ratio of $Z_{n-1}$ and $Z_n$ can asymptotically be expressed as
$$
\xi= \lim\limits_{n\to\infty}\frac{1}{n}\log{\frac{Z_{n-1}}{Z_n}} = -\beta \log\left(\frac{\beta \psi_{\calc}\sigma^2}{2}\right) +\frac{3}{2}\beta \, .
$$
\end{lem}

\begin{proof} 
First, we analyse the partition function $Z_n$. It holds:
\be
\nonumber Z_n = \int\cdots \int\prod\limits_{1\leq i<j\leq n}|\lambda_i^2 -\lambda_j^2|^{\beta} \prod\limits_{j=1}^n \lbrack \lambda_j^{\alpha /n}\exp(-\frac{1}{\psi_{\calc}\sigma^2}\lambda_j^2) \rbrack^n \prod\limits_{j=1}^n d\lambda_j\, ,
\ee
where we integrate over the whole space $\Sigma^n\subseteq \lbrack\, 0,\infty\,\rbrack^n$. 
Now we use the following transformation of Selberg's integral (cf.~\cite[Equation (17.6.5)]{Mehta:book}), which is defined for positive integers $n$: 
\be
\label{selberg-transf}\int\limits_0^{\infty} \cdots\int\limits_0^{\infty} \prod\limits_{1\leq i<j\leq n}|\lambda_i-\lambda_j|^{2\gamma}\prod\limits_{j=1}^n \lambda_j^{\upsilon-1}e^{-\lambda_j}d\lambda_j = \prod\limits_{j=0}^{n-1}\frac{\Gamma(1+\gamma+j\gamma)\Gamma(\upsilon + j\gamma)}{\Gamma(1+\gamma)}\, .
\ee
For our purpose of analysing the asymptotics of the partition function, we substitute in (\ref{selberg-transf}) first $\lambda_j=ay_j$ (for some $a\in\R$) and then $y_j=\frac{\lambda^2_j}{2}$, which leads to
\bea
\nonumber &&\hspace{-12ex}\int\limits_{-\infty}^{\infty} \cdots\int\limits_{-\infty}^{\infty} \prod\limits_{1\leq i<j\leq n}|\lambda_i^2-\lambda_j^2|^{2\gamma}\prod\limits_{j=1}^n |\lambda_j|^{2\upsilon-1}e^{-\frac{a}{2}\lambda_j^2}d\lambda_j \\
\label{selberg-transf-2}&&\hspace{12ex}= \left(\frac{a}{2}\right)^{-\gamma n(n-1)-\upsilon n}2^{-n} \prod\limits_{j=0}^{n-1}\frac{\Gamma(1+\gamma+j\gamma)\Gamma(\upsilon + j\gamma)}{\Gamma(1+\gamma)}\, .
\eea
Note that the integral in (\ref{selberg-transf-2}) is an even function in all parameters $\lambda_j, j=1, \ldots,n$, so we get a combinatorial factor $2^n$ if we integrate over only the positive real axis. We get
\bea
\nonumber &&\hspace{-12ex}\int\limits_0^{\infty} \cdots\int\limits_0^{\infty} \prod\limits_{1\leq i<j\leq n}|\lambda_i^2-\lambda_j^2|^{2\gamma}\prod\limits_{j=1}^n |\lambda_j|^{2\upsilon-1}e^{-\frac{a}{2}\lambda_j 2}d\lambda_j \\
\label{selberg-transf-3}&&\hspace{12ex}= \left(\frac{a}{2}\right)^{-\gamma n(n-1)-\upsilon n}\prod\limits_{j=0}^{n-1}\frac{\Gamma(1+\gamma+j\gamma)\Gamma(\upsilon + j\gamma)}{\Gamma(1+\gamma)}\, .
\eea
In order to get a closed formula for the partition function, we take $\gamma = \frac{\beta}{2}$, $\upsilon=\frac{\alpha}{2}$ and $a=\frac{2n}{\psi_{\calc}\sigma^2}$.
This leads to
\bea
\nonumber Z_n&=&\left(\frac{\psi_{\calc}\sigma^2}{n}\right)^{\frac{\beta}{2}n(n-1)+\frac{\alpha +1}{2}n} \cdot \prod\limits_{j=1}^n\frac{\Gamma(1+\frac{\beta}{2}j)\cdot \Gamma(\frac{\alpha+1}{2}+\frac{\beta}{2}(j-1))}{\Gamma(1+\frac{\beta}{2})}\; \mbox{ and }\\
\nonumber Z_{n-1}&=&\left(\frac{\psi_{\calc}\sigma^2}{n-1}\right)^{\frac{\beta}{2}(n-1)(n-2)+\frac{\alpha +1}{2}(n-1)}\cdot \prod\limits_{j=1}^{n-1}\frac{\Gamma(1+\frac{\beta}{2}j)\cdot \Gamma(\frac{\alpha+1}{2}+\frac{\beta}{2}(j-1))}{\Gamma(1+\frac{\beta}{2})} \, .
\eea
Now we can start the analysis of the behaviour of $\frac{Z_{n-1}}{Z_n}$. We observe
\be
\label{max-rot}\frac{\left(\frac{\psi_{\calc}\sigma^2}{n-1}\right)^{\frac{\beta}{2}(n-1)(n-2)+\frac{\alpha +1}{2}(n-1)}}{\left(\frac{\psi_{\calc}\sigma^2}{n}\right)^{\frac{\beta}{2}n(n-1)+\frac{\alpha +1}{2}n}}= \left(\frac{n}{\psi_{\calc}\sigma^2}\right)^{\beta(n-1)+\frac{\alpha+1}{2}}\left(\frac{n-1}{n}\right)^{\frac{\beta}{2}(n-1)(n-2)+\frac{\alpha+1}{2}(n-1)}\, ,
\ee
which covers the non-Gamma function expressions of $\frac{Z_{n-1}}{Z_n}$.
We now consider the part of $Z_{n-1}/Z_n$ that consists of products of Gamma functions. Most of the Gamma functions cancel out each other; it remains
\be
\label{max-gruen} 
\frac{\Gamma(1+\frac{\beta}{2})}{\Gamma(1+\frac{\beta}{2} n)\Gamma\frac{\alpha +1}{2}+\frac{\beta}{2}(n-1))} \, .
\ee
We apply Stirling's formula to the Gamma expressions in (\ref{max-gruen}) and obtain
$$
\Gamma\left(1+ \frac{\beta}{2}n\right) \simeq  \left(\frac{n\beta}{2e}\right)^{\frac{\beta n}{2}}, \,\,\,
\Gamma\left(\frac{\alpha +1}{2} + \frac{\beta}{2}(n-1)\right) \simeq  \left(\frac{\beta(n-1)}{2e}\right)^{\frac{\beta(n-1)}{2}},
$$
where the equivalence is {\em logarithmic equivalence}, which means for positive numbers $(a_n), (b_n)$ that  $a_n \simeq b_n$ iff $\lim_{n\to \infty}\frac{1}{n}\log a_n / \lim_{n\to \infty}\frac{1}{n}\log b_n = 1$. 
Note that the numerator of (\ref{max-gruen}) is logarithmically equivalent to $1$. This leads to 
Equation (\ref{max-gruen}) being logarithmically equivalent to
\bea
\nonumber (\mbox{\ref{max-gruen}}) &\simeq & \left(\frac{n\beta}{2e}\right)^{-\frac{\beta n}{2}} \left(\frac{\beta(n-1)}{2e}\right)^{-\frac{\beta(n-1)}{2}}\\
\label{max-gelb} &=& \left(\frac{n}{\psi_{\calc}\sigma^2}\right)^{-\frac{2\beta n-\beta}{2}}\left(\frac{\beta \psi_{\calc}\sigma^2}{2e}\right)^{-\frac{2\beta n - \beta}{2}}\left(\frac{n-1}{n}\right)^{-\frac{\beta(n-1)}{2}}\, .
\eea
Now we combine (\ref{max-rot}) and (\ref{max-gelb}) and get
\bea
\nonumber\frac {Z_{n-1}}{Z_n} &\simeq &  \left(\frac{n}{\psi_{\calc}\sigma^2}\right)^{-\frac{\beta}{2}+\frac{\alpha+1}{2}}\left(\frac{\beta \psi_{\calc}\sigma^2}{2e}\right)^{-\frac{2\beta n - \beta}{2}}\left(\frac{n-1}{n}\right)^{\frac{\beta}{2}(n-1)(n-2)+\frac{\alpha+1}{2}(n-1)-\frac{\beta(n-1)}{2}}\, .
\eea
This leads directly to
$$
\lim\limits_{n\to\infty}\frac{1}{n}\log \frac{Z_{n-1}}{Z_n}  = 
 -\beta\log\left(\frac{\beta \psi_{\calc}\sigma^2}{2}\right) +\beta + \frac{\beta}{2}\, ,
$$
where we use 
the rule of de l'Hospital. 
\end{proof}
Obviously $w(x)= \exp(-\frac{x^2}{\psi_{\calc}\sigma^2})$ and hence (a1) and (a2) are fulfilled and we obtain
\bea
I(x) 
\label{5.63} &=&  -\beta\int \log |x^2-y^2| \, d\mu_w(y) + \frac{x^2}{\psi_{\calc}\sigma^2}  -\left(\int \log w(y) \, d\mu_w (y)+\xi\right)\, .
\eea
We define the function
\be
\label{bdg-phi}\Phi(t,\mu) := \int \log |t^2-y^2|d\mu(y) - \frac{t^2}{\psi_{\calc}\sigma^2 \beta}\, .
\ee
Consider $\frac{d}{dt} \Phi(t,\mu_w)$ for $t\geq b_w$ (hence $t^2-y^2 \geq 0$ for $y$ being in the support of $\mu_w$). For convenience, we set $c:=b_w^2 =2\psi_{\calc}\sigma^2 \beta$.
$$
\frac{d}{dt}\Phi(t,\mu_w)= \int \frac{d}{dt} \log (t^2-y^2) d\mu_w(y) - \frac{4t}{c}
= \int\limits_{0}^{\sqrt{c}} \frac{4t}{t^2-y^2}\frac{2}{c \pi}\sqrt{c - y^2}\,  dy - \frac{4t}{c}\, .
$$
The integral is equal to 
$\frac{4}{c\pi}\left( 2t\frac{\pi}{2}- \sqrt{t^2 -c}\pi\right)$ and hence $\frac{d}{dt}\Phi(t,\mu_w)= -\frac{4}{c}\sqrt{t^2-c}$ and by
the fundamental theorem of calculus we get $\Phi(x,\mu_w)= -\int\limits_{\sqrt{c}}^x \frac{4}{c} \sqrt{t^2-c}\, dt + \Phi(\sqrt{c},\mu_w)$.
We calculate $\Phi(\sqrt{c},\mu_w)= \log\left(\frac{\psi_{\calc}\sigma^2\beta}{2}\right) -1$. 
Now it remains to calculate $\int \log w(y)\, d\mu_w(y)$:
$$
\nonumber \int\log w(y)\,d\mu_w(y) = -\frac{4}{c\pi}\int_0^{\sqrt{c}} \frac{y^2}{\psi_{\calc}\sigma^2} \sqrt{c-y^2}\, dy = -\frac{\beta}{2}.
$$
Summarising we obtain the desired rate function (\ref{bdg-ratefunction}).

\subsubsection{The Chiral ensembles}\label{section-chiral}
Although the chiral ensembles and the BdG ensembles are closely related, there are still a few differences between both models. The most noticeable difference is that the empirical measure converges towards a different limiting law. While it converges to something resembling a semicircle law in the BdG case, in the chiral case the limiting measure of the empirical measure behaves Mar\v{c}enko-Pastur-like. Another difference lies in the number of eigenvalues. While for three out of four of the BdG ensembles, we have $n$ different eigenvalues, the chiral ensembles have only $s(n) < n$ different eigenvalues (specified below).
Like \cite[Section 4]{Eichelsbacher/Stolz:2006}, we give the detailed calculations just for the class BDI; the calculations for the other two classes AIII and CII are very similar. Consider $p(n)=s(n)\wedge t(n)$ and $\theta=2$. Assume, without restriction and for simplicity, that $s(n)\leq t(n)$, $n\in \N$,
and take $w_n(x)=x^{\frac{\beta(t(n)-s(n))+\beta -1 }{n}} e^{-\frac{x^2}{2\sigma^2}}$. Denote the partition function by $Z_{n, s(n)}$.
As already seen for the BdG ensembles, we start with a transformation of the Selberg formula, namely equation (\ref{selberg-transf-3}). We use this equation (with an index shift; we start with $j=1$ instead of $j=0$) with parameters $\gamma = \beta/2$, $\upsilon = \frac{\beta}{2}(t(n)-s(n)+1)$, $a=\frac{n}{\sigma^2}$ and $n=s(n)$ and obtain
\bea
\nonumber Z_{n,s(n)} &=& \int\limits_0^{\infty} \cdots\int\limits_0^{\infty} \prod\limits_{1\leq i<j\leq s(n)}|\lambda_i^2-\lambda_j^2|^{\beta}\prod\limits_{j=1}^{s(n)} |\lambda_j|^{\beta(t(n)-s(n))+\beta -1}e^{-\frac{n}{2\sigma^2}\lambda_j^2}d\lambda_j \\[2ex]
\nonumber &= &\left(\frac{n}{2\sigma^2}\right)^{-\frac{\beta}{2} s(n)(s(n)-1)-\frac{\beta}{2}(t(n)-s(n)+1)s(n)}\\
\nonumber && \hspace{12ex}\times\prod\limits_{j=1}^{s(n)}\frac{\Gamma(1+j\frac{\beta}{2})\Gamma(\frac{\beta}{2}(t(n)-s(n)+1) + (j-1)\frac{\beta}{2})}{\Gamma(1+\frac{\beta}{2})}\\[2ex]
\label{5.42} &= &\left(\frac{2\sigma^2}{n}\right)^{\frac{\beta}{2}s(n)(n-s(n))}\prod\limits_{j=1}^{s(n)}\frac{\Gamma(1+j\frac{\beta}{2})\Gamma(\frac{\beta}{2}(n-2s(n)+j)}{\Gamma(1+\frac{\beta}{2})}\, .
\eea
To simplify calculations, from now on we set $s:=s(n)$ and $t:=t(n)$ (recall that $n=t+s$).
We now recall the condition stated in equation (\ref{5.67}), that is, $\limn \frac{w_n(x)}{w_{p(n)-1}(x)}=1$. We choose $s(n-1)=s(n)-1=s-1$ and $t(n-1)= t(n)-1=t-1$ and get
$$
\frac{w_n(x)}{w_{p(n)-1}(x)} =|x|^{\frac{\beta(t-s)+\beta-1}{n}-\frac{\beta(t-1-s+1)+\beta-1}{n-2}} = |x|^{\frac{-2(\beta(t-s)+\beta-1)}{n(n-2)}}\, .
$$
This choice of $s(n-1)$ and $t(n-1)$ assures that the exponent of $x$ converges to $0$, since $\limn s(n)/n =\kappa \in (0,\infty)$. 
Consequently, equation (\ref{5.67}) is fulfilled.
Therefore, we compare $Z_{n,s}$ with $Z_{n-2,s-1}$, where the latter is
\bea
\nonumber Z_{n-2,s-1} &=& \left(\frac{2\sigma^2}{n-2}\right)^{\frac{\beta}{2} ((n-2)-(s-1))(s-1)}\prod_{j=1}^{s-1}\frac{\Gamma(1+j\frac{\beta}{2})\Gamma(\frac{\beta}{2}((n-2)-2(s-1)+j))}{\Gamma(1+\frac{\beta}{2})}\\
\label{5.43} &=& \left(\frac{2\sigma^2}{n-2}\right)^{\frac{\beta}{2} (n-s-1)(s-1)}\prod_{j=1}^{s-1}\frac{\Gamma(1+j\frac{\beta}{2})\Gamma(\frac{\beta}{2}((n-2s+j))}{\Gamma(1+\frac{\beta}{2})}\, .
\eea
We now divide $Z_{n-2,s-1}$ by $Z_{n,s}$, using (\ref{5.42}) and (\ref{5.43}). The first expression in the fraction
\be
\label{5.44} \frac{Z_{n-2,s-1}}{Z_{n,s}}:= B \cdot G
\ee
 is
\be
\nonumber B= \left(\frac{2\sigma^2}{n}\right)^{-\frac{\beta}{2}(n-s)s} \cdot   \left(\frac{2\sigma^2}{n-2}\right)^{\frac{\beta}{2}(n-s-1)(s-1)}\simeq  \left(\frac{2\sigma^2}{n}\right)^{\frac{\beta}{2}(1-n)} = \left(\frac{n}{2\sigma^2}\right)^{\frac{\beta}{2}(n-1)}\, . 
\ee
Now, we have to deal with the Gamma expressions of $Z_{n,s}$ and $Z_{n-2,s-1}$. We get
$$
G = 
\frac{\Gamma(1+\frac{\beta}{2})}{\Gamma(1+\frac{\beta}{2}s)\Gamma(\frac{\beta}{2}(n-s))}.
$$
We now apply Stirling's formula to the Gamma expressions (as for the BdG ensembles) and obtain for (\ref{5.44})
\bea
\nonumber \frac{Z_{n-2,s-1}}{Z_{n,s}}&\simeq& \left(\frac{n}{2\sigma^2}\right)^{\frac{\beta}{2}(n-1)} \left(\frac{\beta s}{2e}\right)^{-\frac{\beta}{2} s}\left(\frac{\beta(n-s)}{2e}\right)^{-\frac{\beta}{2}(n-s)}\\
\label{5.45} &=&  \left(\frac{n-s}{n}\cdot \frac{\beta\sigma^2}{e}\right)^{-\frac{\beta}{2}n}\left(\frac{ s}{n-s}\right)^{-\frac{\beta}{2} s}\left(\frac{n}{2\sigma^2}\right)^{-\frac{\beta}{2}}\, .
\eea
Now, we build the LDP limit and obtain with equation (\ref{5.45})
$$
\limn \one \log \frac{Z_{n-2,s-1}}{Z_{n,s}}
= -\frac{\beta}{2}\left(\log(1-\kappa)+\kappa \log \left(\frac{\kappa}{1-\kappa}\right)+\log (\beta\sigma^2)-1\right)\, ,
$$
where we used that $\limn s/n=\kappa$. With $\limn w_n(x) = x^{\beta (1-2\kappa)}e^{-\frac{x^2}{2\sigma^2}} =w(x)$, the assumptions (a1) and (a2)
are obviously fulfilled.  Therefore, we can apply Corollary \ref{corrlambdamax} on the largest eigenvalue of the chiral ensembles. As rate function, we get
$$
I(x) 
=: -\Phi(x, \mu_w)  - \left(\kappa\int \log w(y) \, d\mu_w (y) +\xi\right)\, ,
$$
with
\be
\label{phi_chiral} \Phi(x,\mu)=\beta\kappa \int\log |x^2-y^2| \, d\mu(y) +  \beta(1-2\kappa) \log x-\frac{x^2}{2\sigma^2} \, .
\ee
As in the BdG case, we look at the derivative of $\Phi(x,\mu_w)$. Note that, for $x\geq b_w$, it holds that $x^2\geq y^2$ because the right endpoint of the support of $\mu_w$, $b_w$, equals $\sqrt{b}$. Therefore, we can omit the absolute values in the logarithm:
\bea
\nonumber \frac{d}{dx}\Phi(x,\mu_w) &=& \beta \kappa \int \frac{d}{dx} \log (x^2-y^2) \, d\mu_w(y) - \frac{x}{\sigma^2} +\frac{\beta (1-2\kappa)}{x}\\
\label{5.53} &=& \frac{1}{\sigma^2 \pi}\int_{\sqrt{a}}^{\sqrt{b}} \frac{2x}{(x^2-y^2)y} \sqrt{(y^2-a)(b-y^2)}\, dy - \frac{x}{\sigma^2} +\frac{\beta (1-2\kappa)}{x}.
\eea
We have 
$\mu_w(y) = 1_{[\sqrt{a},\sqrt{b}]}(y)\frac{1}{\sigma^2\beta\kappa\pi y}\sqrt{(y^2-a)(b-y^2)}$ 
with $a=2\sigma^2\beta \left(\frac{1}{2}-\sqrt{\kappa(1-\kappa)}\right)$ and  $b=2\sigma^2\beta \left(\frac{1}{2}+\sqrt{\kappa(1-\kappa)}\right)$
(see \cite{Eichelsbacher/Stolz:2006}). We now take a closer look at the integral in Equation (\ref{5.53}). Standard calculus (or suitable computeralgebra software -- we used \begin{tt}Mathematica\end{tt}) gives
\begin{equation} \label{5.55}
\int_{\sqrt{a}}^{\sqrt{b}} \frac{2x\sqrt{(y^2 - a)(b - y^2)})}{ (x^2 - y^2)y}\, dy 
= -\frac{\pi}{x}(\sqrt{ab}+\sqrt{(a-x^2)(b-x^2)}+x^2)\, ,
\end{equation}
where $i=\sqrt{-1}$ is the imaginary unit. We used that it holds for the complex logarithm that $\log(-x) = \log x +i\pi$ for all $x\in\R^+$ (this is due to the usual representation of the principal value of the complex logarithm and the fact that $x$ has imaginary component $0$).
Now we insert equation (\ref{5.55}) in equation (\ref{5.53}) and apply the fundamental theorem of calculus on the function $\Phi(x,\mu_w)$ for $b_w\leq t\leq x$. Thus, we get
\bea
\nonumber \Phi(x,\mu_w) &=& \int_{b_w}^x \frac{d}{dt}\Phi(t,\mu_w)\, dt + \Phi(b_w,\mu_w)\\
\nonumber &=&  - \frac{1}{\sigma^2}\int_{b_w}^x\frac{1}{t}(\sqrt{ab}+\sqrt{(a-t^2)(b-t^2)}+t^2)\, dt + \Phi(b_w,\mu_w)\, ,
\eea
where $\Phi(b_w,\mu_w)$ is constant. 
Therefore, we get as a rate function for the BDI ensemble
\be
\label{5.65} I(x)=\frac{1}{\sigma^2}\int_{b_w}^x\frac{1}{t}(\sqrt{ab}+\sqrt{(a-t^2)(b-t^2)}+t^2)\, dt + c\, ,
\ee
where $c<\infty$ is a constant. Again it follows $P(\lambda_n^* \to b_w) =1$.

\section{Large deviations for multiple orthogonal ensembles}
Multiple orthogonal polynomials are a generalisation of orthogonal polynomials in which the orthogonality
is distributed among a number of orthogonality weights. They appear in random matrix theory in the form of special determinantal
point processes that are called multiple orthogonal polynomial (MOP) ensembles. In \cite{Kuijlaars:2010, Kuijlaars:2010b}
the appearance of MOP in a variety of random matrix models and models related with particles
following non-intersecting paths have been considered. 
To a finite number of weight functions $w_1, \ldots, w_p$ on $\Bbb R$ and a multi-index $\vec{n} =(n_1, \ldots, n_p) \in {\Bbb N}^p$ we
associate a monic polynomial $P_{\vec{n}}$ of degree $n:= |\vec{n}| := n_1 + \cdots + n_p$ such that
$$
\int_{-\infty}^{\infty} P_{\vec{n}}(x) x^k w_j(x) \, dx = 0, \quad \text{for} \,\, k=0, \ldots, n_j-1, \,\, j=1, \ldots, p.
$$
If $P_{\vec{n}}$ uniquely exists then it is called the multiple orthogonal polynomial (MOP) associated with the weights $w_1, \ldots, w_p$
and multi-index $\vec{n}$. In \cite{Kuijlaars:2010} the following result was presented. Assume that
\begin{equation} \label{condMOP}
\frac{1}{Z_n} \det [f_j(x_k)]_{j,k=1, \ldots,n} \biggl[ \prod_{1 \leq j < k \leq n} (x_k-x_j) \biggr]
\end{equation}
is a probability density function on ${\Bbb R}^n$, where the linear span of $f_1, \ldots, f_n$ is the same as the linear span
of $\{x^k w_j(x) | k=0, \ldots, n_j-1, \, j=1, \ldots, p\, \}$. Then the MOP exists and is given by
$$
P_{\vec{n}}(x) = \E \biggl[ \prod_{j=1}^n(x-x_j) \biggr],
$$
where the expectation is taken with respect to the p.d.f \eqref{condMOP}, which can be interpreted as the expectation
of the random polynomial $\prod_{j=1}^n (x-x_j)$ with roots $x_1, \ldots, x_n$ from a determinantal point process
on the real line. The p.d.f \eqref{condMOP} is called a {\it MOP ensemble}. It was first observed in \cite{Bleher/Kuijlaars:2004} that
random matrix models with an external source lead naturally to MOP ensembles. 
The weights $w_1, \ldots, w_p$ are an {\it Angelesco system} if there are disjoint intervals $\Gamma_1, \ldots, \Gamma_p \subset {\Bbb R}$, such
that $\supp (w_j) \subset \Gamma_j$, $j=1, \ldots, p$. In the Angelesco case, $\det [f_j(x_k)]_{j,k=1, \ldots,n}$ is of block form and results in
$$
\det [f_j(x_k)]_{j,k=1, \ldots,n} = \prod_{i=1}^p \biggl( \Delta(X^{(i)}) \prod_{k=1}^{n_i} w_i(x_k^{(i)}) \biggr)
$$
with $x_k^{(i)} := x_{N_{i-1}+k} \in \Gamma_i$, $N_i=\sum_{j=1}^i n_j$ (with $N_0=0$) and $X^{(i)} = (x_1^{(i)}, \ldots, x_{n_i}^{(i)})$ and
$$
\Delta(X) = \prod_{1 \leq j <k \leq n} (x_k-x_j) \quad \text{for} \,\, X=(x_1, \ldots,x_n),
$$
the Vandermonde determinant. Thus an Angelesco system gives rise to a MOP ensemble, the {\it Angelesco ensemble}, and the joint
p.d.f is
\begin{equation} \label{angelesco}
\frac{1}{Z_n} \prod_{i=1}^p \Delta(X^{(i)})^2 \prod_{1 \leq i < j \leq p} \Delta(X^{(i)}, X^{(j)}) \prod_{i=1}^p \prod_{k=1}^{n_i} w_i(x_k^{(i)}),
\end{equation}
where 
$$
\Delta(X,Y) := \prod_{k=1}^n \prod_{j=1}^m (x_k-y_j)
$$
for $X=(x_1, \ldots,x_n)$ and $Y=(y_1, \ldots, y_m)$.
We now consider the situation that $|\vec{n}| = n \to \infty$ and $n_j \to \infty$ for every $j=1, \ldots, p$ in such a way that
\begin{equation} \label{c1} 
\frac{n_j}{n} \to r_j \quad \text{for} \,\, j=1, \ldots, p
\end{equation}
with $0 < r_j <1$ and $\sum_{j=1}^p r_j =1$. Let us consider varying weights 
\begin{equation} \label{varying}
w_i(x)= e^{-n V_i(x)} 
\end{equation} 
for any $i=1, \ldots, p$. Denote by $\lambda_j^* := \max_{1 \leq k \leq n_j}$ the
the $j$-{\it th maximal eigenvalue} for every $j=1, \ldots, p$. We will study the asymptotic behaviour of $(\lambda_1^*, \ldots, \lambda_p^*)$.

\begin{theorem}[LDP for Angelesco ensembles] \label{ldpMOP}
Assume that every weight function $w_i$ in \eqref{varying} satisfies assumption (a1) and (a2) and assume that \eqref{c1}
is fulfilled. Assume moreover assumption \eqref{assZ} for $Z_n$. Then the sequence 
$(\lambda_1^*, \ldots, \lambda_p^*)_n$ satisfies a
LDP on $\R^p$ with speed $n$ and good rate function
\begin{eqnarray} \label{ratean}
I(x_1, \ldots, x_p) & = & \sum_{i=1}^p r_i^2 \int \log \left|x_i - y\right|^{-2} \mu_i^*(dy) \\
&& + \sum_{1 \leq i < j \leq p}  r_i r_j \int \log \left|x_i - y\right|^{-1} \mu_j^*(y) + \sum_{i=1}^p r_i V_i(x_i) - \zeta, \nonumber
\end{eqnarray}
where 
$$
\zeta :=  \kappa \sum_{i=1}^p \int \log w_i(y) d \mu_i^*(dy) + \xi.
$$
Here $\mu^* =(\mu_1^*, \ldots, \mu_p^*)$ is assumed to be a unique minimiser could in Theorem \cite[Section III]{Eichelsbacher/Sommerauer/Stolz:2011}
\end{theorem}

\begin{remark}
We would also be able to consider {\it Nikishin ensembles} with $p \geq 2$ weights, see \cite{Kuijlaars:2010} and references therein.
This is because the determinantal structure of the joint density of the eigenvalues \cite[(4.14)]{Kuijlaars:2010} consists
of Vandermonde-like products. Nikishin interaction
arises in the asymptotic analysis of eigenvalues of {\it banded Toeplitz matrices} as well as in a {\it two-matrix model}, 
see \cite[Section 5.4]{Kuijlaars:2010}.
\end{remark}

\section{Proofs}
This section is devoted to the proof of Theorem \ref{maintheoremlambdamax}. Theorem \ref{ldpMOP} will not be proved since the arguments
are very similar, but technically much more involved.

\subsection{Exponential tightness}
\begin{lemma} \label{5.1}
We have
\be
\nonumber\lim_{M\to\infty}\limsupn \one \pe_n(\lambdamax \geq M) = -\infty \, .  
\end{equation}
\end{lemma}

For the proof of Lemma \ref{5.1}, we need the following technical inequality.

\begin{lemma}\label{5.8} For all $|x|\geq \max(b_w,1)$ and for all $\lambda \in \Sigma$, there exists a constant $c <\infty$ with
\bea
\label{5.2}|x-\lambda||x^{\theta}-\lambda^{\theta}|w_n(\lambda)\leq c\cdot |x|^{\theta+1}\, .
\eea
\end{lemma}

\begin{proof}
The proof is inspired by the work of \cite{Feral:2008}.  First we show that there exist
$\eps > 0$, $T> 0$ and $n_0\in \N$ such that for all $|\lambda|\geq T$ and for all  $n \geq n_0$ we have
\begin{equation}
\label{5.5} \log w_n(\lambda) \leq - \frac{(\theta+1)(\kappa+\eps)}{2}\log(1+\lambda^2).
\end{equation}
Due to assumption (a2) for unbounded $\Sigma$, there exists an $n_0\in\N$ such that
$\lim_{\lambda\to\pm\infty}|\lambda|^{(\theta + 1)(\kappa + \eps)}\sup_{n\geq n_0}w_n(\lambda) = 0$. 
This implies
\be
\nonumber \lim_{\lambda\to\pm\infty} \left(\frac{(\theta+1)(\kappa+\eps)}{2}\log(1+ \lambda^2) + \log \sup_{n\geq n_0} w_n(\lambda) \right)=-\infty
\ee
and, since the first summand is positive, it follows that $\exists\, n_1\in\N$ and $\exists\, T\in \R$ such that for all $n\geq \max (n_0,n_1)$ and for all $|\lambda|\geq T$:
\be
\label{5.4} - \log\sup_{n\geq n_0} (w_n(\lambda)) \geq \frac{(\theta+1)(\kappa+\eps)}{2}\log(1+ \lambda^2)\, .
\ee
Since the logarithm is a monotonic increasing function, we can cancel the supremum in (\ref{5.4}). It follows \eqref{5.5} for $n\geq \max (n_0,n_1)$ and for $|\lambda|\geq T$.
Note that this implies that for $n\geq \max (n_0,n_1)$ and for $|\lambda|\geq T$ that $w_n(\lambda)\leq 1$.
To be complete, we have to cover also the case $|\lambda|\leq T$. For $|\lambda| \leq T$ and $|x|\geq 1$ we have
$$
\nonumber |x-\lambda||x^{\theta}-\lambda^{\theta}| w_n(\lambda) \leq \sup_{|\lambda|\leq T} w_n(\lambda)(|x|+T)(|x|^{\theta}+T^{\theta})\leq c\cdot |x|^{\theta+1}\, ,
$$
with $c>\infty$ being a constant. Now we consider the left-hand side of inequality (\ref{5.2}) and start with a case distinction:
For $|x|\geq |\lambda|$, we have
\bea
\nonumber |x-\lambda||x^{\theta}- \lambda^{\theta}|w_n(\lambda)&=& |x-\lambda||x^{\theta}- \lambda^{\theta}|\exp\left(\log w_n(\lambda)\right)\\
\label{5.6} &\leq & |x-\lambda||x^{\theta}- \lambda^{\theta}|\exp\left( - \frac{(\theta+1)(\kappa+\eps)}{2}\log (1+\lambda^2)\right)\\
\nonumber &\leq & |x-\lambda||x^{\theta}- \lambda^{\theta}| \leq 4 |x||x^{\theta}| \leq  4|x|^{\theta + 1},
\eea
where we used (\ref{5.5}) for the first inequality in (\ref{5.6}) and for the second inequality that $\frac{(\theta+1)(\kappa+\eps)}{2}> 0$, $\log (1+\lambda^2)\geq 0$. 
The third inequality is due to the fact that $|x^{\theta}-\lambda^{\theta}|\leq |x^{\theta}| + |\lambda^{\theta}|\leq 2|x^{\theta}|$ for $|x|\geq |\lambda|$. The last inequality uses $\theta\in\N$.
For $|x| < |\lambda|$, we start again with inequality (\ref{5.6}) and obtain:
\bea
\nonumber |x-\lambda||x^{\theta}- \lambda^{\theta}|w_n(\lambda)&\leq & |x-\lambda||x^{\theta}- \lambda^{\theta}|\exp\left( - \frac{(\theta+1)(\kappa+\eps)}{2}\log (1+\lambda^2)\right)\\
\nonumber &\leq & 4|\lambda|^{\theta + 1}\exp\left( - \frac{(\theta+1)(\kappa+\eps)}{2}\log (1+\lambda^2)\right)\\
\nonumber &=& 4 \exp \left((\theta+1)\log|\lambda| - \frac{(\theta+1)(\kappa+\eps)}{2}\log (1+\lambda^2)\right)\, .\\
\label{5.7} 
\eea
Now, since $|\lambda|\geq |x|\geq \max(b_w,1)$ and since $c:= \frac{(\theta+1)(\kappa+\eps)}{2}> 0$, the term $-c\log(1+\lambda^2)$ dominates the exponent in (\ref{5.7}). Therefore, the whole exponent is bounded by $1$ and we obtain (\ref{5.7}) $\leq 4\cdot 1 \leq 4|x|^{\theta + 1}$.
\end{proof}

\begin{proof} (Proof of Lemma \ref{5.1})

We take a closer look at the following density:
\bea
\nonumber \pe_n(\lambdamax \geq M)& = &\frac{1}{Z_n} \int\limits_M^{\infty}\int\limits_{\Sigma^{p(n)-1}} \prod_{j=1}^{p(n)} w_n(\lambda_j)^n \prod_{1\leq i<j\leq p(n)} |\lambda_i-\lambda_j||\lambda_i^{\theta}-\lambda_j^{\theta}|\prod_{j=1}^{p(n)} \, d\lambda_j\\
\nonumber &=& \frac{1}{Z_n} \int\limits_M^{\infty}\int\limits_{\Sigma^{p(n)-1}} w_n(\lambda_1)^n \prod_{j=1}^{p(n)}|\lambda_1 - \lambda_j||\lambda_1^{\theta}-\lambda_j^{\theta}| \cdot \prod_{j=2}^{p(n)} \frac{w_n(\lambda_j)^n}{w_{n-1}(\lambda_j)^{n-1}}\\
\nonumber && \qquad\cdot \prod_{j=2}^{p(n)} w_{n-1} (\lambda_j)^{n-1} \prod_{2\leq i<j\leq p(n)} |\lambda_i-\lambda_j||\lambda_i^{\theta}-\lambda_j^{\theta}| \prod_{j=2}^{p(n)}\, d\lambda_j d\lambda_1\\
\nonumber &=& \frac{Z_{n-1}}{Z_n}\int\limits_M^{\infty}\int\limits_{\Sigma^{p(n)-1}} w_n(\lambda_1)^n \prod_{j=2}^{p(n)}|\lambda_1 - \lambda_j||\lambda_1^{\theta}-\lambda_j^{\theta}| \\
 && \qquad \qquad \qquad \quad  \times \prod_{j=2}^{p(n)} \frac{w_n(\lambda_j)^n}{w_{n-1}(\lambda_j)^{n-1}} \, d\pe_{n-1}(\lambda_2\ldots\lambda_{p(n)}) d \lambda_1\, , \label{5.3}
\eea
where we added $\prod_{j=2}^{p(n)}\frac{w_{n-1}(\lambda_j)^{n-1}}{w_{n-1}(\lambda_j)^{n-1}}$ in the first step and substituted $\pe_{n-1}$ in the second step.
Now, we know that we can write the last product in (\ref{5.3}) as
\be
\nonumber \prod_{j=2}^{p(n)} \frac{w_n(\lambda_j)^n}{w_{n-1}(\lambda_j)^{n-1}} =\prod_{j=2}^{p(n)} \left(\frac{w_n(\lambda_j)}{w_{n-1}(\lambda_j)}\right)^{n-1}\prod_{j=2}^{p(n)}w_n(\lambda_j)\, .
\ee
Therefore, we get for equation (\ref{5.3}) that
\bea
\nonumber \pe_n(\lambdamax \geq M) &=& \frac{Z_{n-1}}{Z_n}\int\limits_M^{\infty}\int\limits_{\Sigma^{p(n)-1}} w_n(\lambda_1)^n \prod_{j=2}^{p(n)}w_n(\lambda_j)|\lambda_1 - \lambda_j||\lambda_1^{\theta}-\lambda_j^{\theta}| \\
\label{5.70} &&\qquad \qquad \quad \times \prod_{j=2}^{p(n)} \left(\frac{w_n(\lambda_j)}{w_{n-1}(\lambda_j)}\right)^{n-1} \, d\pe_{n-1}(\lambda_2\ldots\lambda_{p(n)}) d \lambda_1\\[3ex]
\nonumber &\leq & \tilde{c}\cdot \frac{Z_{n-1}}{Z_n} \int\limits_M^{\infty}\int\limits_{\Sigma^{p(n)-1}} w_n(\lambda_1)^n |\lambda_1|^{(\theta+1)(p(n)-1)}\\
 &&\qquad \qquad\quad\times\prod_{j=2}^{p(n)} \left(\frac{w_n(\lambda_j)}{w_{n-1}(\lambda_j)}\right)^{n-1} \, d\pe_{n-1}(\lambda_2\ldots\lambda_{p(n)}) d \lambda_1\, , \label{5.9}
\eea
where we used Lemma \ref{5.8} for the inequality with constant $\tilde{c}=c^{p(n)-1}$.
Now, we take a closer look at a part of the integrand of (\ref{5.9}):
\begin{equation} \label{5.10}
\nonumber w_n(\lambda_1)^n |\lambda_1|^{(\theta+1)(p(n)-1)} 
=  \left\lbrack \left(w_n(\lambda_1)|\lambda_1|^{(\theta +1)(\frac{p(n)}{n} +\eps)}\right)\cdot\left( |\lambda_1|^{-(\theta +1)(\eps +\one)}\right)\right\rbrack^n.
\end{equation}
For large $n$, the fraction $\frac{p(n)}{n}$ is near its limiting value $\kappa$. For big $|\lambda_1|$, we know from assumption (a2) that the first bracket of equation (\ref{5.10}) is small. 
Therefore for large $n$ and large $|\lambda_1|$ we have:
$(\ref{5.10})\leq  |\lambda_1|^{-n(\theta +1)(\eps +\one)}$.
Since we assumed in equation (\ref{5.67}) that $\limn \frac{w_n(\lambda_j)}{w_{n-1}(\lambda_j)}=1$, we have
$\limn\prod_{j=2}^{p(n)}\left(\frac{w_n(\lambda_j)}{w_{n-1}(\lambda_j)}\right)^{n-1} = 1$.
Therefore, (\ref{5.9}) becomes 
\begin{eqnarray}
\label{5.71} (\ref{5.9}) &\simeq& \tilde{c} \cdot \frac{Z_{n-1}}{Z_n}   \int\limits_M^{\infty} |\lambda_1|^{-n(\theta +1)(\eps +\one)} \, d\lambda_1 \,\int\limits_{\Sigma^{p(n)-1}} \, d\pe_{n-1}(\lambda_2\ldots\lambda_{p(n)}) \nonumber \\
\label{5.11}  &=& \tilde{c} \cdot \frac{Z_{n-1}}{Z_n}   \int\limits_M^{\infty} |\lambda_1|^{-n(\theta +1)(\eps +\one)} \, d\lambda_1
\end{eqnarray}
for large values of $n$ and $|\lambda_1|$. The last step is due to the fact that $\pe_{n-1}$ is a probability.
Now we integrate (\ref{5.11}) out with respect to $\lambda_1$ and get
\be
\nonumber (\ref{5.11}) = \tilde{c}\cdot \frac{Z_{n-1}}{Z_n}\cdot  \frac{1}{1-n(\theta+1)(\eps +\one)}M^{1-n(\theta + 1)(\eps +\one)}\, .
\ee
Now we take the LDP-limit and obtain
\bea
\nonumber \limsupn \one \log \pe_n(\lambdamax \geq M) &\leq& \limsupn \one \log \tilde{c} + \limsupn \one \log \frac{Z_{n-1}}{Z_n}\\
\nonumber && + \limsupn\one \log \frac{1}{1-n(\theta +1)(\eps +\one)} \\
\nonumber && + \limsupn\one\log M^{1-n(\theta+1)(\eps +\one)}\, .
\eea
The first summand is zero, the second equals $\xi$ due to assumption \eqref{assZ}. The third summand has size $O(\frac{\log n}{n})$ and thus converges to $0$ for $n\to\infty$, and the last summand is
\begin{equation}
\nonumber \limsupn \one \log M^{1-n(\theta+1)(\eps +\one)} = \limsupn \frac{1-n(\theta+1)(\eps +\one)}{n}\log M = -\hat{c}\log M
\end{equation}
for a positive constant $\hat{c}$. Therefore, since $\xi$ is a constant, we obtain the result.
\end{proof}
Since the density of the eigenvalues behaves symmetrically for $x\to \infty$ and $x\to -\infty$, we can analogously prove the following lemma:
\begin{lem}\label{lambdamin}
Define $\lambdamin := \min_{j=1}^{p(n)}\lambda_j$. Then we have
\be
\nonumber \lim_{M\to\infty} \limsupn \one \log \pe_n(\lambdamin < -M) = -\infty \, .
\ee
\end{lem}

\subsection[$I$ is a good rate function]{{\boldmath $I$} is a good rate function}\label{section-good-rate-function}
\begin{lem}
$I(x)$ is a good rate function.
\end{lem}

\begin{proof} Since we assumed in (a1) that $w(x)$ is continuous, $I(x)$ is continuous in $x$ on the interval $(b_w, \infty)$ and lower semicontinuous in $x=b_w$.  For any $x,y \in\R$ it holds $\log |x-y| \leq \log (|x|+1)+ \log (|y|+1)$. Hence for any probability measure $\mu$
\bea
\nonumber && \hspace{-15ex}\int \log |x-y|+\log |x^{\theta}-y^{\theta}|d\mu(y)\\
\nonumber  &\leq& \int\log(|x|+1)+\log(|y|+1) +\log(|x^{\theta}|+1)+\log(|y^{\theta}|+1)\,d\mu(y)\\
\nonumber &=& \int \log (|x|+1)(|x^{\theta}|+1)\, d\mu(y) + \int \log (|y|+1)(|y^{\theta}|+1)\, d\mu(y) \\
\label{5.24} &=& \log (|x|+1)(|x^{\theta}|+1) + \int \log (|y|+1)(|y^{\theta}|+1)\, d\mu(y)\, .
\eea
With (\ref{5.24}) it follows for the rate function
\bea
\nonumber I(x) &=& -\kappa \int\log |x-y| +\log |x^{\theta}-y^{\theta}| \,d\mu(y) -\log w(x) - c\\
\nonumber &\geq & -\kappa \log (|x|+1)(|x^{\theta}|+1) - \kappa \int\log (|y|+1)(|y^{\theta}|+1)\, d\mu(y) -\log w(x) - c\, .\\
\label{5.25}
\eea
Again, we neglect those terms that are constant concerning $x$ (i.e.~the integral and $c$ in (\ref{5.25})). We also omit the absolute values, since we are interested in large positive values of $x$. We apply inequality (\ref{5.5}), which gives information about the behaviour of $w(x)$ for $x\to\infty$, and obtain for $x> \max\lbrace b_w,1\rbrace$ that
\bea
\nonumber  -\kappa \log (x+1)(x^{\theta}+1) -\log w(x) &\geq & -\kappa \log (x+1)(x^{\theta}+1) + \frac{(\theta+1)(\kappa+\eps)}{2} \log(1+x^2)\\
\nonumber &\geq & -\kappa\log x^{(\theta+1)}+ \log  \frac{(\theta+1)(\kappa+\eps)}{2} \log(1+x^2)\\
\nonumber &=& \calo\left(\log(x^{-\kappa(\theta+1)})\right) + \calo \left(x^{(\theta+1)(\kappa+\eps)}\right)\\
\nonumber &=& \calo\left(\log x^{\eps(\theta +1)}\right)\, .
\eea
This converges to $\infty$ for $x\to \infty$ since $\eps >0$ and $\theta \in \N$.
Therefore, $I(x)$ is strictly increasing on $\lbrack b_w,\infty)$, it has compact level sets and is thus a good rate function.
\end{proof}

\subsection{Proof of a weak LDP}\label{lambdamax-weak-ldp}
In this section, we prove a weak LDP for $\lambdamax$ (for a definition see \cite[Section 1.2]{Dembo/Zeitouni:LargeDeviations}).
Since $I$ is a good rate function, it suffices to show that for any $x\leq b_w$,
\be
\label{weak-ldp-1}\limsupn\one\log \pe_{n}(\lambdamax\leq x)= -\infty
\ee
and for any $x> b_w$,
\be
\label{weak-ldp-2}\limn\one\log \pe_{n}(\lambdamax\geq x)= -I(x)\, ,
\ee
where $b_w$ is the right endpoint of the support of the limiting measure.
The reason we need to prove (\ref{weak-ldp-1}) and (\ref{weak-ldp-2}) lies in the following consideration:
\cite[Theorem 4.1.11]{Dembo/Zeitouni:LargeDeviations} states that for the proof of a weak LDP, we do not need to show the LDP upper/lower bound for every compact/open set belonging to the space $\Sigma$. Instead, it suffices to look at a base of the topology of $\Sigma\subseteq \R$. We choose as a basis the set of open and closed intervals $\cala:=\left\lbrace \lbrack x,y\rbrack, (x,y) \, :\, x<y\right\rbrace$.
Still, we need to verify 
\be
\label{5.32} I(x) = \sup_{\lbrace A \in\cala : x\in A\rbrace}\left\lbrack -\limsupn\one\log \pe_n(\lambdamax \in A)\right\rbrack\, ,
\ee
which is, for closed sets $A=\lbrack x,y\rbrack$ and for all $x<y$, equivalent to
\be
\nonumber -\inf_{t\in\lbrack x,y\rbrack } I(t) = \limsupn\one\log \pe_n(\lambdamax \in \lbrack x,y\rbrack)\, .
\ee
But this is easy to see with equations (\ref{weak-ldp-1}), (\ref{weak-ldp-2}) and the fact that $I$ is increasing on $\lbrack b_w, \infty)$:
We briefly distinguish three cases: if $x<y\leq b_w$, then $\limsupn\one\log\pe_n(\lambdamax \in \lbrack x,y\rbrack) = -\infty$ due to equation (\ref{weak-ldp-1}). If $x\leq b_w\leq y$, we use the fact that in this interval lies the ``typical'' value of $\lambdamax$, therefore $\limsupn\one\log\pe_n(\lambdamax \in\lbrack x,y\rbrack) = 0$. Finally, if $b_w\leq x<y$, we get with equation (\ref{weak-ldp-2}) that $\limsupn\one\log \pe_n(\lambdamax \in\lbrack x,y\rbrack)= - \inf_{t\in\lbrack x,y\rbrack} I(t)$.
With the continuity of $I$, we get that we have the same limits if we apply the same calculations to open intervals $(x,y)$ instead of closed intervals $\lbrack x,y\rbrack$. 
Therefore, we obtain a weak LDP for $\lambdamax$ as soon as we have proved equations (\ref{weak-ldp-1}) and (\ref{weak-ldp-2}), which is what we do in the rest of this section.

\noindent
We first prove (\ref{weak-ldp-1}).  We know from \cite[Theorem 2.1]{Eichelsbacher/Sommerauer/Stolz:2011}, that the empirical measure $L_n=\one\sum\limits_{j=1}^n\delta_{\lambda_j}$ of the eigenvalues of a biorthogonal random matrix obeys a large deviations principle with speed $n^2$ and a good rate function.
Since $I$ is a good rate function, it achieves its minimum.
Now we fix $x< b_w$. Since $\lambdamax\leq x < b_w$ (this means that there is no eigenvalue of $X_n$ in the interval $(x, b_w\rbrack$), we find a bounded continuous function $f\in \calc_b(\Sigma)$ with $\int f\, dL_n = 0$ but $\int f \, d\mu_w > 0$. Therefore, we conclude with \cite[Theorem 2.1]{Eichelsbacher/Sommerauer/Stolz:2011} that the probability of $\lbrace \lambdamax \leq x\rbrace$ is exponentially decaying (with speed $n^2$) and therefore it follows equation (\ref{weak-ldp-1}).
We split up the proof of Equation (\ref{weak-ldp-2}) into two steps. We first show that $-I(x)$ is an upper bound of the scaled probability on the left-hand-side of (\ref{weak-ldp-2}), then we prove that $-I(x)$ is also a lower bound.

\subsubsection{Proof of the upper bound}
Observe that we have, for any $M> x\geq b_w$, the following inequality:
\be
\label{5.26} \pe_n(\lambdamax \geq x)\leq \pe_n(\lambdamax \in\lbrack x, M\rbrack) + \pe_n(\lambdamax > M) \, .
\ee
Since we proved that large values of $\lambdamax$ are exponentially negligible, we choose $M$ large enough, i.e.~that the first term in the right hand side of inequality (\ref{5.26}) is exponentially small, and also that the same holds for the minimal eigenvalue $\lambdamin$, see Lemma \ref{lambdamin}.
Hence we only need to deal with the probability that $\lambdamax \in \lbrack x,M\rbrack$.
Further it holds
\bea
\nonumber && \kappa \int \log|x-y| +\log |x^{\theta}-y^{\theta}| d\mu_w(y) + \log w(x) + \kappa\int\log w(y)d\mu_w(y) \\
&&\quad > \sup_{t\in\left\lbrack M,\infty\right)}\left( \kappa\int \log|t-y| +\log |t^{\theta}-y^{\theta}| d\mu_w(y) + \log w(t) + \kappa\int\log w(y)d\mu_w (y) \right)\, .\nonumber \\
\label{5.27}
\eea
which we will apply later on in the proof. Since the eigenvalues $\lambda_j, 1\leq j \leq p(n),$ are exchangeable, it holds $\pe_n(\lambdamax \in \lbrack x,M\rbrack) \leq n \pe_n(\lambda_1 \in \lbrack x, M\rbrack, |\lambda_j|\leq M \, \forall j\geq 2)$. Now we consider
\bea
\nonumber && \hspace{-9ex}\pe_n(\lambda_1 \in \lbrack x, M\rbrack, |\lambda_j | \leq M \, \forall j\geq 2)\\
\nonumber &=& \frac{1}{Z_n}\int_x^M \int_{I_M} \prod_{1\leq i< j\leq p(n)}|\lambda_i-\lambda_j||\lambda_i^{\theta}-\lambda_j^{\theta}|\prod_{j=1}^{p(n)}w_n(\lambda_j)^n \, d\lambda_2 \cdots d\lambda_{p(n)}\, d\lambda_1\\
\nonumber &=& \frac{Z_{n-1}}{Z_n} \int_x^M \int_{I_M} w_n(\lambda_1)^n \prod_{j=2}^{p(n)} \lbrack w_n(\lambda_j)|\lambda_1-\lambda_j||\lambda_1^{\theta}-\lambda_j^{\theta}|\rbrack \cdot \prod_{j=2}^{p(n)} \left(\frac{w_n(\lambda_j)}{w_{n-1}(\lambda_j)}\right)^{n-1} \, \\
&& \qquad \qquad \qquad \qquad \qquad \qquad \qquad \qquad \qquad \qquad  \times d\pe_{n-1}(\lambda_2,\ldots, \lambda_{p(n)})\, d\lambda_1 \, ,\label{5.12}
\eea
with $I_M:= \left(\Sigma \cap \left(-M, M\right\rbrack\right)^{p(n)-1}$, and where we replaced $\pe_{n-1}$ as in equation (\ref{5.3}).
Now we define, for $t\in\lbrack -M,M\rbrack$ and $\mu$ supported on $\lbrack -M, M\rbrack$, the function
\be
\label{phi_n} \Phi_n(t,\mu):= \int \log |t-y| + \log |t^{\theta}-y^{\theta}|d\mu(y) + \frac{n}{p(n)-1}\log w_n(t) + \int\log w_n(y)\, d\mu(y)\, .
\ee
We consider the first part of the integrand in equation (\ref{5.12}):
\bea
\nonumber && \hspace{-7ex} w_n(\lambda_1)^n \prod_{j=2}^{p(n)} \left\lbrack w_n(\lambda_j)|\lambda_1-\lambda_j||\lambda_1^{\theta}-\lambda_j^{\theta}|\right\rbrack\\
\nonumber &=&  \exp \big\{ n \log w_n(\lambda_1) + \sum_{j=2}^{p(n)}\log w_n(\lambda_j) + \sum_{j=2}^{p(n)}\lbrack \log |\lambda_1 - \lambda_j| + \log |\lambda_1^{\theta}-\lambda_j^{\theta}|\rbrack\big\} \\
\nonumber &=& \exp \left\lbrace (p(n)-1)\left\lbrack\int\log|\lambda_1-y|+\log|\lambda_1^{\theta}-y^{\theta}|\, dL_{n-1}(y) + \frac{n}{p(n)-1}\log w_n(\lambda_1)\right.\right.\\
\nonumber &&\qquad\qquad\qquad\qquad\qquad \left. \left. + \int\log w_n(\lambda_j)\, d L_{n-1}(y) \right\rbrack \right\rbrace \\
\label{5.21} &=& \exp \lbrace (p(n)-1)\Phi_n(\lambda_1,L_{n-1})\rbrace \, ,
\eea
where $L_{n-1}:=\frac{1}{n-1}\sum\limits_{j=2}^{p(n)} \delta_{\lambda_j}$ is the empirical measure of the $p(n)-1$ eigenvalues $\lambda_2,\ldots, \lambda_{p(n)}$. In the third line, we used the identities 
$\sum\limits_{j=2}^{p(n)}\log|\lambda_1^{\theta}-\lambda_j^{\theta}| =(p(n)-1)\int \log|\lambda_1^{\theta}-y^{\theta}|\, dL_{n-1}(y)$ and 
$\sum\limits_{j=2}^{p(n)}\log w_n(\lambda_j) = (p(n)-1)\int \log w_n(y) \, dL_{n-1}(y)$.
Therefore, with equations (\ref{phi_n}) and (\ref{5.21}), we obtain for Equation (\ref{5.12})
\bea
\nonumber && \hspace{-9ex}\pe_n(\lambda_1 \in \lbrack x, M\rbrack, |\lambda_j | \leq M \, \forall j\geq 2)\\
\nonumber &=& \frac{Z_{n-1}}{Z_n}\int\limits_x^M \int\limits_{I_M} e^{(p(n)-1)\Phi_n(\lambda_1,L_{n-1})}\prod_{j=2}^{p(n)} \left(\frac{w_n(\lambda_j)}{w_{n-1}(\lambda_j)}\right)^{n-1} \, d\pe_{n-1}(\lambda_2,\ldots,\lambda_{p(n)})\, d\lambda_1\\
\nonumber &=& \frac{Z_{n-1}}{Z_n}\int\limits_x^M \int\limits_{I_M} e^{  (p(n)-1)\Phi_n(\lambda_1,L_{n-1}) + (n-1)\sum\limits_{j=2}^{p(n)} \log\left(\frac{w_n(\lambda_j)}{w_{n-1}(\lambda_j)}\right) }
\, d\pe_{n-1}(\lambda_2,\ldots,\lambda_{p(n)})\, d\lambda_1\\
\nonumber &=& \frac{Z_{n-1}}{Z_n}\int\limits_x^M \int\limits_{I_M} e^{  (p(n)-1)\left(\Phi_n(\lambda_1,L_{n-1}) +(n-1)\int\log\left(\frac{w_n(y)}{w_{n-1}(y)}\right) dL_{n-1}(y)\right)}\\
\nonumber && \qquad\qquad\qquad\qquad\qquad\qquad\qquad\qquad\qquad\qquad \cdot\,  d\pe_{n-1}(\lambda_2,\ldots, \lambda_{p(n)})\, d\lambda_1\\
\label{5.13} &=&\frac{Z_{n-1}}{Z_n}\int\limits_x^M \int\limits_{\calm_1(\Sigma\cap \lbrack -M,M\rbrack)
} e^{  (p(n)-1)\tilde{\Phi}_n(\lambda_1,\mu)} \, d(\pe_{n-1}\circ L_{n-1}^{-1})(\mu)d \lambda_1\, ,
\eea
where
\be
\label{phitilde}\tilde{\Phi}_n(t,\mu):=\Phi_n(t,\mu)+(n-1)\int\log\left(\frac{w_n(y)}{w_{n-1}(y)}\right) \mu(dy)\, .
\ee
Now, in equation (\ref{5.13}), we split the domain of integration of the inner integral.
\bea
\nonumber\calm_1(\Sigma \cap \lbrack -M, M\rbrack)& =& \left\lbrack \calm_1(\Sigma \cap \lbrack -M, M\rbrack)\cap B(\mu_w,\delta)\rbrack \right.\\
\nonumber &&\qquad\qquad \qquad\qquad \qquad\left.\cup \lbrack \calm_1(\Sigma \cap \lbrack -M, M\rbrack)\cap (B(\mu_w,\delta))^C\right\rbrack\\
\nonumber &=& B_M(\mu_w,\delta)\cup \lbrack \calm_1(\Sigma \cap \lbrack -M, M\rbrack)\cap (B(\mu_w,\delta))^C\rbrack\, .
\eea
This leads to (\ref{5.13}) being equal to
\bea
\nonumber && \hspace{-8ex}\pe_n(\lambda_1 \in \lbrack x, M\rbrack, |\lambda_j | \leq M \, \forall j\geq 2)\\
\nonumber &=& \frac{Z_{n-1}}{Z_n} \left\lbrack \int\limits_x^M \int\limits_{B_M(\mu_w,\delta)} e^{ (p(n)-1)\tilde{\Phi}_n(\lambda_1,\mu)} \, d(\pe_{n-1}\circ L_{n-1}^{-1})(\mu)d \lambda_1\right. \\
\nonumber && \qquad\qquad \left. + \int\limits_x^M \int\limits_{\calm_1(\Sigma\cap \lbrack -M,M\rbrack)\cap (B_M(\mu_w, \delta))^C} e^{ (p(n)-1)\tilde{\Phi}_n(\lambda_1,\mu)} \, d(\pe_{n-1}\circ L_{n-1}^{-1})(\mu)d \lambda_1 \right\rbrack\\
\label{5.28} &=:& \frac{Z_{n-1}}{Z_n} (I_1+I_2)\, .
\eea
We first discuss the integral $I_2$, which (as we shall see) converges to $-\infty$ on an LDP scale. To show this, we prove an upper bound for $\tilde{\Phi}$ and integrate over a larger set than that appearing in the definition of $I_2$.
We choose $\lambda_1\in \lbrack x, M\rbrack$ and $\mu\in\calm_1 (\Sigma\cap\lbrack -M,M\rbrack)$. Then we have
$\log |\lambda_1 - \lambda_j| \leq \log (2M)$ and $\log |\lambda_1^{\theta} - \lambda_j^{\theta}|\leq \log (2M^{\theta})$.
Therefore, for $\Phi_n$ as defined in (\ref{phi_n}),
\be
\nonumber \Phi_n(\lambda_1,\mu) \leq \log (2M) + \log (2M^{\theta}) + \frac{n}{p(n)-1} \sup_{t\in \lbrack x, M\rbrack} \log w_n(t) + \sup_{t\in\lbrack -M,M\rbrack} \log w_n (t)\, ,
\ee
and therefore we have for $\tilde{\Phi}_n$, as defined in (\ref{phitilde}),
\bea
\nonumber \tilde{\Phi}_n(\lambda_1,\mu) &\leq& \log (2M) + \log (2M^{\theta}) + \frac{n}{p(n)-1} \sup_{t\in \lbrack x, M\rbrack} \log w_n(t)\\
\nonumber && \qquad\qquad\qquad + \sup_{t\in\lbrack -M,M\rbrack} \log w_n (t) + (n-1)\sup_{t\in\lbrack -M,M\rbrack} \log \frac{w_n(t)}{w_{n-1}(t)}\, .
\eea
Thus, we have an upper bound for $\tilde{\Phi}_n$ that is independent of $\lambda_2,\ldots,\lambda_{p(n)}$ and that we can place before the inner integral. Further, we have for the domain of integration of the inner integral of $I_2$
that $\lbrace \calm_1 (\Sigma\cap \lbrack -M, M\rbrack)\cap (B(\mu_w,\delta))^C\rbrace \subseteq \lbrace (B(\mu_w,\delta))^C\rbrace$.
This leads to
\begin{eqnarray*}
I_2 & \leq & \exp \bigg\{ \log 2M+\log 2M^{\theta}+\frac{n}{p(n)-1}\sup_{t\in\lbrack x,M\rbrack}\log w_n(t)+\sup_{t\in\lbrack -M,M\rbrack}\log w_n(t)  \\
&& +(n-1)\sup_{t\in\lbrack -M,M\rbrack}\log\frac{w_n(t)}{w_{n-1}(t)} \bigg\} \cdot (\pe_{n-1}\circ L_{n-1}^{-1})(\mu\notin B(\mu_w,\delta)).
\end{eqnarray*}
Now it follows that
\bea
\nonumber I_2 &=& 2M\cdot 2M^{\theta} \sup_{t\in\lbrack x,M\rbrack} w_n(t)^{\frac{n}{p(n)-1}}\sup_{t\in\lbrack -M,M\rbrack} w_n(t)\sup_{t\in\lbrack -M,M\rbrack}\left(\frac{w_n(t)}{w_{n-1}(t)}\right)^{(n-1)}\\
\nonumber && \qquad\qquad\qquad\qquad \qquad\qquad\qquad\qquad\times (\pe_{n-1}\circ L_{n-1}^{-1})(\mu\notin B(\mu_w,\delta))\\[2ex]
\label{5.14} &\leq& 4 M^{(\theta+1)(1+(\kappa+\eps)\frac{n}{p(n)-1})+2}\cdot e^{-c(\delta)n^2}
\eea
for some $M, n$ large enough and for some positive constant $c(\delta)$.
The second step is due to the fact that we exchanged the logarithm and the supremum for $w_n$. This is allowed for $n$ large enough because we are on a compactly supported interval where $w_n$ converges uniformly to the continuous function $w$. The third step uses assumption (a2), the boundedness of $\lbrack x,M\rbrack$ and~$\lbrack -M,M\rbrack$ as well as the fact that $\left(\pe_{n-1}\circ L_{n-1}^{-1}\right)_{n-1}$ obeys a LDP with speed $n^2$, see \cite[Theorem 2.1]{Eichelsbacher/Sommerauer/Stolz:2011}.
Since the order of the first factor in (\ref{5.14}) is smaller than $e^{n^2}$, the estimation is complete because now $I_2$ vanishes on the LDP scale. Remember that the LDP for $\lambdamax$, which we want to prove, has speed $n$.
So we only have to deal with the first integral of equation (\ref{5.28}) and, of course, the factor $Z_{n-1}/Z_n$, applying \cite[Lemma 1.2.15]{Dembo/Zeitouni:LargeDeviations}.
As an aside, we also get with \cite[Lemma 1.2.15]{Dembo/Zeitouni:LargeDeviations}
that it suffices to consider $\pe_n(\lambdamax \in\lbrack x,M\rbrack)$ for the calculation of $\limsupn\one\log\pe_n(\lambdamax \geq x)$, cf.~inequality (\ref{5.26}) and the fact that in the LDP scale, $\pe_n(\lambdamax > M)$ converges to $-\infty$.
Therefore we have:
\bea
\nonumber \limsupn\one\log \pe_n(\lambdamax \geq x)&=& \limsupn\one\log\pe_n(\lambdamax\in \lbrack x,M\rbrack)\\
\nonumber &\leq &\limsupn \one \log (\frac{Z_{n-1}}{Z_n} I_1)\\
\label{5.29} &\leq& \xi + \limsupn\one\log I_1\, .
\eea
The last line is due to assumption \eqref{assZ}.
Now we shall take a closer look at the integral $I_1$. It holds for a probability measure $\tilde{\mu}\in \calm_1(\Sigma)$ with $\tilde{\mu}\in B_M(\mu_w,\delta)$ that $\tilde{\Phi}_n(\lambda_1,L_{n-1})\leq \sup_{\tilde{\mu}\in B_{M}(\mu_w,\delta)} \tilde{\Phi}_n(\lambda_1,\tilde{\mu})$,
which is independent of $\lambda_2,\ldots, \lambda_{p(n)}$. Hence we obtain
\bea
\nonumber I_1 &\leq & \int\limits_x^M e^{(p(n)-1)\sup_{\tilde{\mu}\in B_M(\mu_w,\delta)}\tilde{\Phi}_n(\lambda_1,\tilde{\mu})}\, d \lambda_1\, \int\limits_{B_M(\mu_w,\delta)}  \, d(\pe_{n-1}\circ L_{n-1}^{-1})(\mu)\\
\nonumber &\leq &  \int\limits_x^M e^{(p(n)-1)\sup_{\tilde{\mu}\in B_M(\mu_w,\delta)}\tilde{\Phi}_n(\lambda_1,\tilde{\mu})}\, d \lambda_1\\
\label{5.15} &\leq& \exp \lbrace(p(n)-1) \sup_{\substack{\mu\in B_M(\mu_w,\delta)\\ t\in \lbrack x,M\rbrack}} \tilde{\Phi}_n(t,\mu) \rbrace \cdot (M-x)\, ;
\eea
the second step follows because $\pe_{n-1}\circ L_{n-1}^{-1}$ is a probability measure. In the third step we  build the supremum over all allowed values of $\lambda_1$, extract the integrand (which is now independent of the integration variable) out of the integral and integrate out over $1$.
Now we look at $I_1$ on the LDP scale. It follows with inequality (\ref{5.15}):
\bea
\nonumber \limsupn\one\log I_1 &\leq& \limsupn\frac{(p(n)-1)}{n} \sup_{\substack{\mu\in B_M(\mu_w,\delta)\\ t\in \lbrack x,M\rbrack}} \tilde{\Phi}_n(t,\mu)  +\limsupn \one\log(M-x)\\
\label{5.16} &=& \kappa \sup_{\substack{\mu\in B_M(\mu_w,\delta)\\\lambda_1\in \lbrack x,M\rbrack}} \limsupn \tilde{\Phi}_n(\lambda_1,\mu)\, .
\eea
Now we step back and consider $\tilde{\Phi}_n$ for $n\to\infty$. On compact intervals, $\log w_n$ converges uniformly to $\log w$ due to assumption (a1). Therefore, if we insert this in the definition of $\tilde{\Phi}_n$, (\ref{phitilde}), we get
\bea
\nonumber \limsupn \tilde{\Phi}_n(\lambda_1,\mu) &=& \int \log|\lambda_1-y| +\log |\lambda_1^{\theta}-y^{\theta}| d\mu(y) + \kappa^{-1} \log w(\lambda_1)\\
\nonumber && \qquad\qquad + \int\log w(y)d\mu(y) + \limsupn \int \log\left(\frac{w_n(y)}{w_{n-1}(y)}\right)^{n-1}d\mu(y)\, ,
\eea
where the last summand converges to $0$ for $n\to\infty$.
We define as limit of $\Phi_n$ and $\tilde{\Phi}_n$ the new function
\be
\label{phi} \Phi(\lambda_1,\mu) =: \int \log|\lambda_1-y| +\log |\lambda_1^{\theta}-y^{\theta}| d\mu(y) + \kappa^{-1} \log w(\lambda_1) + \int\log w(y)d\mu(y) 
\ee
and obtain, if we insert this in (\ref{5.16}),
\be
\nonumber \limsupn\one \log I_1 \leq \kappa\sup_{\substack{\mu\in B_M(\mu_w,\delta)\\\lambda_1\in\lbrack x,M\rbrack}} \Phi(\lambda_1,\mu)\, .
\ee
By (\ref{5.29}), we have
\be
\nonumber \limsupn\one\log\pe_n(\lambdamax \geq x)\leq \xi + \limsupn\one\log I_1\, .
\ee
and therefore
\be
\label{5.31} \limsupn \one\log\pe_n(\lambdamax \geq x)\leq \xi + \kappa \lim_{\delta \searrow 0} \sup_{\substack{\mu\in B_M(\mu_w,\delta)\\\lambda_1\in\lbrack x,M\rbrack}} \Phi(\lambda_1,\mu)\, .
\ee
Note that the calculation for the scaled and logarithmised integral $I_2$ is also independent of $\delta$, since the calculations there hold for all $\delta >0$.
Since the measures $\mu$ that are inserted in $\Phi(x,\mu)$ come from a probability space endowed with a topology that is compatible with weak convergence, we have
\be
\label{5.30} \lim_{\delta \searrow 0} \sup_{\substack{\mu\in B_M(\mu_w,\delta)\\\lambda_1\in\lbrack x,M\rbrack}} \Phi(\lambda_1,\mu)= \sup_{\lambda_1\in\lbrack x,M\rbrack} \Phi(\lambda_1,\mu_w) = \sup_{\lambda_1\in\lbrack x,\infty)} \Phi(\lambda_1,\mu_w) \, .
\ee
The last equality is due to (\ref{5.27}).
We now look at the relation between the function $\Phi$ and the rate function $I$ from Theorem \ref{maintheoremlambdamax}. It holds for all $t\geq b_w$
$\Phi(t,\mu_w) = -\kappa^{-1} (I(t)+\xi)$.
Since we proved that $I(x)$ is strictly increasing on $\lbrack b_w,\infty)$ and continuous on $(b_w,\infty)$, we have that:
$\Phi(x,\mu_w)$ is continuous for $x > b_w$ (and lower semicontinuous for $x=b_w$)
and $\Phi(x,\mu_w)$ is strictly decreasing for $x\geq b_w$.
Therefore, the supremum in equation (\ref{5.30}) is attained in the leftmost value that $\lambda_1$ can attain, namely, $x$. That implies, inserted in (\ref{5.31}), for all $x\geq b_w$
\be
\nonumber \limsupn\one\log \pe_n(\lambdamax \geq x) \leq \xi + \kappa \Phi(x,\mu_w)\, .
\ee
This is the desired upper bound we wanted to prove.

\subsection{Proof of the lower bound}
In the last step of the proof of Theorem \ref{maintheoremlambdamax}, we prove that our rate function $-I$ is a lower bound for the scaled probability on the left-hand-side of equation (\ref{weak-ldp-2}).
We fix $y>x>r>b_w$ and $\delta> 0$. Since $\lbrace \lambda_1\in\lbrack x,y\rbrack, \max_{j=2}^{p(n)}|\lambda_j|\leq r\rbrace\subset\lbrace \lambdamax \in\lbrack x,M\rbrack\rbrace$ for large $M$ and $y$ close to $x$, we look at the following probability:
\be
\label{5.17} \pe_n(\lambda_1\in\lbrack x,y\rbrack, \max_{j=2}^{p(n)}|\lambda_j|\leq r) = \frac{Z_{n-1}}{Z_n} \int_x^y\int_{\calm_1(\Sigma\cap \lbrack -r,r\rbrack)} e^{ (p(n)-1)\tilde{\Phi}_n(\lambda_1,\mu)} d(\pe_{n-1}\circ L_{n-1}^{-1})(\mu)d\lambda_1\, 
\ee
(compare with (\ref{5.13})).
For $\lambda_1 \in \lbrack x,y\rbrack$, the integrand has the following lower bound:
\be
\nonumber \exp \left((p(n)-1)\tilde{\Phi}_n(\lambda_1,\mu)\right) \geq  \exp \left((p(n)-1)\inf_{t\in\lbrack x,y\rbrack}\tilde{\Phi}_n(t,\mu)\right)\, .
\ee
This leads to a lower bound for (\ref{5.17}):
\bea
\nonumber (\ref{5.17}) &\geq & \frac{Z_{n-1}}{Z_n} \int_x^y \, d\lambda_1\int_{\calm_1(\Sigma\cap \lbrack -r,r\rbrack)} e^{(p(n)-1)\inf_{t\in\lbrack x,y\rbrack}\tilde{\Phi}_n(t,\mu)} d(\pe_{n-1}\circ L_{n-1}^{-1})(\mu)\\
\nonumber &\geq &\frac{Z_{n-1}}{Z_n} (y-x)\int_{B_r(\mu_w,\delta)} e^{(p(n)-1)\inf_{t\in\lbrack x,y\rbrack}\tilde{\Phi}_n(t,\mu)} d(\pe_{n-1}\circ L_{n-1}^{-1})(\mu)\\
\nonumber &\geq & \frac{Z_{n-1}}{Z_n} (y-x) \exp \bigg\{(p(n)-1)\inf_{\substack{t\in\lbrack x,y\rbrack\\\mu\in{B_r(\mu_w,\delta)}}}\tilde{\Phi}_n(t,\mu) \bigg\} \pe_{n-1}(L_{n-1}\in B_r(\mu_w,\delta))\, ,
\eea
where in the second step we built the second integral over the smaller set
$\lbrace L_{n-1}\in B_r(\mu_w,\delta)\rbrace \subset \lbrace L_{n-1} \in \calm_1(\Sigma\cap \lbrack -r,r\rbrack)\rbrace$
and integrated out the first integral. In the third step, we once more made the integrand independent of the integration variables and put it outside the integral.
Now, we build the LDP limit:
\bea
\nonumber &&\hspace{-4ex}\liminfn\one\log \pe_n(\lambda_1\in\lbrack x,y\rbrack, \max_{j=2}^{p(n)}|\lambda_j|\leq r)\\
\nonumber &&\qquad \geq \, \liminfn \one\log \frac{Z_{n-1}}{Z_n} + \liminfn \one \log (y-x) + \liminfn \frac{p(n)-1}{n}\inf_{\substack{t\in\lbrack x,y\rbrack\\\mu\in B_r(\mu_w,\delta)}} \tilde{\Phi}_n (t,\mu)\\
\label{5.18} &&\qquad\qquad + \liminfn \one\log \pe_{n-1}(L_{n-1}\in B_r(\mu_w,\delta))\\
\label{5.19} &&\qquad =  \xi + \liminfn \frac{p(n)-1}{n}\inf_{\substack{t\in\lbrack x,y\rbrack\\\mu\in B_r(\mu_w,\delta)}} \tilde{\Phi}_n (t,\mu)\, ,
\eea
where we used assumption \eqref{assZ} and the fact that the second and the last summand of inequality (\ref{5.18}) are zero. Observe that for the last term this follows from the large deviations result for the empirical measure that provides $\pe_{n-1}(L_{n-1}\not\in B(\mu_w,\delta))\to 0$, therefore $\pe_{n-1}(L_{n-1}\in B(\mu_w,\delta))\to 1$ and it holds $\log \pe_{n-1}(L_{n-1}\in B(\mu_w,\delta))\to 0$.
We just look at the second summand of (\ref{5.19}) and fill in the definition of $\tilde{\Phi}_n$ as defined in equation (\ref{phitilde}):
\bea
\nonumber && \hspace{-4ex}\liminfn \frac{p(n)-1}{n}\inf_{\substack{t\in\lbrack x,y\rbrack\\\mu\in B_r(\mu_w,\delta)}} \tilde{\Phi}_n (t,\mu)\\
\nonumber && = \liminfn \frac{p(n)-1}{n}\inf_{\substack{t\in\lbrack x,y\rbrack\\\mu\in B_r(\mu_w,\delta)}} \left(\int\log |t-y|+\log|t^{\theta}-y^{\theta}|\, d\mu(y)+\frac{n}{p(n)-1}\log w_n(t) \right.\\
\nonumber && \left.\qquad \qquad\qquad \qquad \qquad \qquad + \int\log w_n(y)d\mu(y)+ (n-1)\int\log\left(\frac{w_n(y)}{w_{n-1}(y)}\right)d\mu(y)\right)\\
\nonumber && \geq \kappa \inf_{\substack{t\in\lbrack x,y\rbrack\\\mu\in B_r(\mu_w,\delta)}} \left(\int\log |t-y| +\log |t^{\theta} -y^{\theta}|d\mu(y) \right.\\
\label{5.20}  && \left.\qquad \qquad \qquad\qquad \qquad \qquad \qquad  + \kappa^{-1}\log w(t) + \int\log w(y)d\mu(y)\right)\\
\nonumber && = \kappa \inf_{\substack{t\in\lbrack x,y\rbrack\\\mu\in B_r(\mu_w,\delta)}} \Phi(t,\mu)\, .
\eea
For inequality (\ref{5.20}), we used again (as in the proof of the upper bound) the uniform convergence of $\log w_n$ to $\log w$ on compact sets as assumed in assumption (a1) and exchanged the Limes inferior and the infimum. The last equality is just the definition of $\Phi$, see equation (\ref{phi}).
Therefore we have
\be
\nonumber \liminfn \one\log\pe_n(\lambda_1\in\lbrack x,y\rbrack, \max_{j=2}^{p(n)}|\lambda_j|\leq r) \geq \xi + \kappa \liminfn\one\log(\inf_{\substack{t\in\lbrack x,y\rbrack\\\mu\in B_r(\mu_w,\delta)}} \Phi(t,\mu))\, .
\ee
Now we again build the limit $\delta \searrow 0$ on both sides, which does not affect all terms but one, and observe that the function $\Phi$ is continuous on $\lbrack x,y\rbrack \times \calm_1(\lbrack -r,r\rbrack)$ for $y>x>r>b_w$. We also let $y\searrow x$ and obtain
\be
\nonumber \liminfn\one \log \pe_n(\lambda_1\in\lbrack x,y\rbrack, \max_{j=2}^{p(n)}|\lambda_j|\leq r) \geq \kappa \Phi(x,\mu_w)+\xi\, ,
\ee
which is the desired lower bound. We therefore proved Theorem \ref{maintheoremlambdamax}. 

\section{Appendix}
Borodin \cite{Borodin:1999} starts with a set of $n$ real random variables $\lambda_1,\ldots,\lambda_n$ with values in the interval $I\subseteq \R$ and with joint probability density function
\begin{equation}
\label{bordensity} q_n(\lambda_1,\ldots,\lambda_n) = \frac{1}{Z_n} \det(\eta_i(\lambda_j))_{i,j=1,\ldots,n}\det(\xi_i(\lambda_j))_{i,j=1,\ldots,n}\cdot \prod_{j=1}^n w(\lambda_j)\, ,
\end{equation}
where $Z_n$ is the normalisation constant and $\eta$ and $\xi$ are real-valued functions.
We suppose that we are able to biorthogonalise the families $\lbrace\eta_i\rbrace$ and $\lbrace\xi_i\rbrace$ regarding the weight function $w$, i.e.~that we have an inner product $\langle\xi, \eta\rangle = \int_I \xi(\lambda)\eta(\lambda) w(\lambda)\, d\lambda$.
This entails that we find two other families of functions, $\lbrace p_i(\lambda)\rbrace_{i=1}^n$ and $\lbrace q_i(\lambda)\rbrace_{i=1}^n$ where the $p_i$ and $q_i$ are monic (i.e.~the leading term has prefactor $+1$) polynomials in $\lambda$ of degree $i-1$, and that are orthogonal (not necessarily orthonormal) with respect to the weight function, i.e.~we have $p_i \in \spann \lbrace \xi_1,\ldots,\xi_n\rbrace$, $q_i \in\spann\lbrace\eta_1,\ldots,\eta_n\rbrace$ and  
$$
h_j \delta_{i,j} :=\langle p_i,q_j\rangle = \int_I p_i(\lambda)q_j(\lambda) w(\lambda) \, d\lambda \qquad \forall i,j\, .
$$
For the construction of such polynomials, see e.g.~\cite[Section V.A]{Lueck/Sommers/Zirnbauer:2006}, where the authors used a modified Gram-Schmidt algorithm to obtain the desired polynomials. Define now the matrix $g$ with
\be
\nonumber g_{i,j}:= \int_I \eta_i(\lambda)\xi_j(\lambda) w(\lambda)\, d\lambda\, ,
\ee
and assume it to be not singular. Now we have everything at hand to formulate two propositions that help to calculate the normalising constant of biorthogonal ensembles, see  \cite[Proposition 5.5]{Forrester:book} and \cite[Equation (24)]{Desrosiers/Forrester:2008}.
 
\begin{propos}\label{propos-zn-detg} 
For the normalisation constant of general biorthogonal ensembles with density \eqref{bordensity},
we have
\bea
\label{clever} Z_n = n! \det g =  c \cdot n! \prod_{j=0}^{n-1} h_j,
\eea
where $c$ is a constant.
\end{propos}

\newcommand{\SortNoop}[1]{}\def\cprime{$'$} \def\cprime{$'$}
  \def\polhk#1{\setbox0=\hbox{#1}{\ooalign{\hidewidth
  \lower1.5ex\hbox{`}\hidewidth\crcr\unhbox0}}}
\providecommand{\bysame}{\leavevmode\hbox to3em{\hrulefill}\thinspace}
\providecommand{\MR}{\relax\ifhmode\unskip\space\fi MR }
\providecommand{\MRhref}[2]{%
  \href{http://www.ams.org/mathscinet-getitem?mr=#1}{#2}
}
\providecommand{\href}[2]{#2}

\end{document}